\documentclass{vldb}
\usepackage{balance,array,microtype,pifont}
\usepackage[bf]{subfigure}

\usepackage{color, times, mathptmx, amsmath, amssymb, amsthm, verbatim, 
  mdwlist, graphicx, footnote, multirow, xspace, tikz, colortbl,dcolumn, mathtools,zi4}

\usepackage[noadjust]{cite}

\DeclareMathAlphabet{\mathcal}{OMS}{cmsy}{m}{n}
\let\footnotesize\small

\usepackage{algorithm}

\usepackage[noend]{algpseudocode}

\algrenewcommand{\alglinenumber}[1]{\small#1:}

\theoremstyle{definition}
\newtheorem{lemma}{Lemma}

\makeatletter
\def\thm@space@setup{\thm@preskip=.5em
\thm@postskip=.5em}
\makeatother

\newtheorem{definition}{Definition}
\newtheorem{theorem}{Theorem}

\theoremstyle{remark}

\usepackage{thmtools}
\declaretheoremstyle[%
  spaceabove=10pt,%
  spacebelow=10pt,%
  postheadspace=1em,%
  qed=\qedsymbol%
]{mystyle} 
\declaretheorem[name={Claim},style=mystyle]{claim}

\usepackage{scrextend}
\usepackage[hyphens]{url}

\usepackage[hidelinks, pdftitle={Coordination Avoidance in Database Systems},
            pdfauthor={Peter Bailis, Alan Fekete, Michael J. Franklin, Ali Ghodsi, Joseph M. Hellerstein, Ion Stoica}]{hyperref}
\usepackage{breakurl}

\newcommand{\suchthat}{\textrm{ s.t. }}
\newcommand{\mathand}{\textrm{ and }}

\newcommand{\minihead}[1]{{\vspace{.45em}\noindent\textbf{#1.} }}
\newcommand{\miniheadit}[1]{{\vspace{.45em}\noindent\textit{#1.} }}

\newif\ifextended
\extendedtrue

\newcommand{\appexperiment}{A}
\newcommand{\apptheory}{B}
\newcommand{\appapply}{C}

\newcommand{\rappendix}[1]{{\ifextended Appendix #1\else\cite[Appendix
    #1]{ca-extended}\fi }}

\newcommand{\rsection}[1]{{\ifextended Section #1\else\cite[Section #1]{ca-extended}\fi }}

\newcommand{\ttf}[1]{\texttt{#1}\xspace}

\newcommand{\cfree}{coordination-free\xspace}
\newcommand{\cfreedom}{coordination-freedom\xspace}

\newcommand{\fullnameconfluence}{invariant confluence\xspace}
\newcommand{\iconfluent}{$\mathcal{I}$-confluent\xspace}
\newcommand{\iconfluence}{$\mathcal{I}$-confluence\xspace}

\newcommand{\dpc}{\ttf{D-2PC}}
\newcommand{\cpc}{\ttf{C-2PC}}

\usepackage{etoolbox}
\newcommand{\zerodisplayskips}{%
  \setlength{\abovedisplayskip}{4pt}
  \setlength{\belowdisplayskip}{4pt}
  \setlength{\abovedisplayshortskip}{4pt}
  \setlength{\belowdisplayshortskip}{4pt}}
\appto{\normalsize}{\zerodisplayskips}
\appto{\small}{\zerodisplayskips}
\appto{\footnotesize}{\zerodisplayskips}

\usepackage[T1]{fontenc}

\newenvironment{introenumerate}
{

   \vspace{-.5em}
   \newcounter{qdcounter}
    \begin{list}{\arabic{qdcounter}.~}{\usecounter{qdcounter}\leftmargin=1em}
        \setlength{\topsep}{0em}
        \setlength{\parskip}{0pt}
        \setlength{\partopsep}{0pt}
        \setlength{\parsep}{0pt}         
        \setlength{\itemsep}{.5em} 
        \setlength{\itemindent}{0em}
}
{
    \end{list} 
    \vspace{-.5em}
}

\usepackage[leftmargin=0em]{quoting}

\begin{document}
%
\conferenceinfo{XXX}{YYY}

\title{Coordination Avoidance in Database Systems\ifextended\\(Extended Version)\fi}

{\author{Peter Bailis, Alan Fekete{\fontsize{12}{14}$^\dagger$},
    Michael J. Franklin, Ali Ghodsi, Joseph M. Hellerstein, Ion Stoica
    \\{\affaddr{UC Berkeley and {\fontsize{12}{14}$^\dagger$}University of Sydney}}}}
\maketitle

\begin{abstract}
  Minimizing coordination, or blocking communication between
  concurrently executing operations, is key to maximizing scalability,
  availability, and high performance in database systems. However,
  uninhibited coordination-free execution can compromise application
  correctness, or consistency. When is coordination necessary for
  correctness? The classic use of serializable transactions is
  sufficient to maintain correctness but is not necessary for all
  applications, sacrificing potential scalability. In this paper, we
  develop a formal framework, \fullnameconfluence, that determines whether an
  application requires coordination for correct execution. By
  operating on application-level invariants over database states
  (e.g., integrity constraints), \fullnameconfluence analysis provides a
  necessary and sufficient condition for safe, coordination-free
  execution. When programmers specify their application invariants,
  this analysis allows databases to coordinate only when anomalies
  that might violate invariants are possible. We analyze the
  \fullnameconfluence of common invariants and operations from real-world
  database systems (i.e., integrity constraints) and applications and
  show that many are invariant confluent and therefore achievable without
  coordination. We apply these results to a proof-of-concept
  coordination-avoiding database prototype and demonstrate sizable
  performance gains compared to serializable execution, notably a
  25-fold improvement over prior TPC-C New-Order performance on a 200
  server cluster.
\end{abstract}

\section{Introduction}
\label{sec:intro}


Minimizing coordination is key in high-performance, scalable database
design. Coordination---informally, the requirement that concurrently
executing operations synchronously communicate or otherwise stall in
order to complete---is expensive: it limits concurrency between
operations and undermines the effectiveness of scale-out across
servers. In the presence of partial system failures, coordinating
operations may be forced to stall indefinitely, and, in the
failure-free case, communication delays can increase
latency~\cite{hat-vldb,gilbert-cap}. In contrast, coordination-free
operations allow aggressive scale-out,
availability~\cite{gilbert-cap}, and low latency
execution~\cite{pacelc}. If operations are coordination-free, then
adding more capacity (e.g., servers, processors) will result in additional
throughput; operations can execute on the new resources without
affecting the old set of resources. Partial failures will not affect
non-failed operations, and latency between any database replicas can
be hidden from end-users.


Unfortunately, coordination-free execution is not always safe. Uninhibited
coordination-free execution can compromise application-level
correctness, or consistency.\footnote{Our use of the term
  ``consistency'' in this paper refers to \textit{application-level} correctness, as
  is traditional in the database
  literature~\cite{gray-virtues,bernstein-book,eswaran-consistency,traiger-tods,davidson-survey}. As
  we discuss in Section~\ref{sec:bcc-practice}, replicated data
  consistency (and isolation~\cite{adya-isolation,hat-vldb})
  models like linearizability~\cite{gilbert-cap} can be cast as
  application criteria if desired.} In canonical banking application
examples, concurrent, coordination-free withdrawal operations can
result in undesirable and ``inconsistent'' outcomes like negative account
balances---application-level anomalies that the database should
prevent. To ensure correct behavior, a database system must coordinate
the execution of these operations that, if otherwise executed
concurrently, could result in inconsistent application state.

This tension between coordination and correctness is
evidenced by the range of database concurrency control policies. In
traditional database systems, serializable isolation provides
concurrent operations (transactions) with the illusion of executing in
some serial order~\cite{bernstein-book}. As long as individual
transactions maintain correct application state, serializability
guarantees correctness~\cite{gray-virtues}. However, each pair of
concurrent operations (at least one of which is a write) can
potentially compromise serializability and therefore will require
coordination to execute~\cite{hat-vldb,davidson-survey}. By isolating
users at the level of reads and writes, serializability can be overly
conservative and may in turn coordinate more than is strictly
necessary for
consistency~\cite{lamport-audit,tamer-book,ic-survey,weihl-thesis}.
For example, hundreds of users can safely and simultaneously retweet
Barack Obama on Twitter without observing a serial ordering of updates
to the retweet counter. In contrast, a range of widely-deployed weaker
models require less coordination to execute but surface read and write
behavior that may in turn compromise
consistency~\cite{dynamo,optimistic,adya-isolation,hat-vldb}. With
these alternative models, it is up to users to decide when weakened
guarantees are acceptable for their
applications~\cite{consistency-borders}, leading to confusion
regarding (and substantial interest in) the relationship between
consistency, scalability, and
availability~\cite{hat-vldb,pacelc,dynamo,gilbert-cap,davidson-survey,kohler-commutativity,redblue-new,queue}.


In this paper, we address the central question inherent in this
trade-off: when is coordination strictly necessary to maintain
application-level consistency? To do so, we enlist the aid of
application programmers to specify their correctness criteria in the
form of \textit{invariants}. For example, our banking application
writer would specify that account balances should be positive (e.g.,
by schema annotations), similar to constraints in modern databases
today. Using these invariants, we formalize a \textit{necessary} and
sufficient condition for invariant-preserving and coordination-free
execution of an application's operations---the first such condition we
have encountered. This property---\fullnameconfluence
(\iconfluence)---captures the potential scalability and availability
of an application, independent of any particular database
implementation: if an application's operations are \iconfluent, a
database can correctly execute them without coordination. If
operations are not \iconfluent, coordination is required to
guarantee correctness. This provides a basis for \textit{coordination
  avoidance}: the use of coordination only when necessary.

While coordination-free execution is powerful, are any \textit{useful}
operations safely executable without coordination? \iconfluence
analysis determines when concurrent execution of specific operations
can be ``merged'' into valid database state; we accordingly analyze
invariants and operations from several real-world databases and
applications. Many production databases today already support
invariants in the form of primary key, uniqueness, foreign key, and
row-level check constraints~\cite{kemme-si-ic,hat-vldb}. We analyze
these and show many are \iconfluent, including forms of foreign key
constraints, unique value generation, and check constraints, while
others, like primary key constraints are, in general, not. We also
consider entire \textit{applications} and apply our analysis to the
workloads of the OLTPBenchmark suite~\cite{oltpbench}. Many of the 
operations and invariants are \iconfluent. As an extended case study, we examine the
TPC-C benchmark~\cite{tpcc}, the preferred standard for evaluating new
concurrency control
algorithms~\cite{abadi-vll,jones-dtxn,calvin,hstore,oltpbench}. We
show that ten of twelve of TPC-C's invariants are \iconfluent under
the workload transactions and, more importantly, compliant TPC-C can
be implemented without any synchronous coordination across servers. We
subsequently scale a coordination-avoiding database prototype
linearly, to over 12.7M TPC-C New-Order transactions per second on
$200$ servers, a 25-fold improvement over prior results.


Overall, \iconfluence offers a concrete grasp on the challenge of
minimizing coordination while ensuring application-level correctness. In seeking a necessary and sufficient (i.e., ``tight'')
condition for safe, coordination-free execution, we require the
programmer to specify her correctness criteria. If either these criteria or
application operations are unavailable for inspection, users must fall
back to using serializable transactions or, alternatively, perform the
same ad-hoc analyses they use today~\cite{queue}. Moreover, it is already well
known that coordination is required to prevent several read/write
isolation anomalies like non-linearizable
operations~\cite{gilbert-cap,hat-vldb}. However, when users
\textit{can} correctly specify their application correctness criteria
and operations, they can maximize scalability without requiring
expertise in the milieu of weak read/write isolation
models~\cite{hat-vldb,adya-isolation}. We have also found that
\iconfluence to be a useful design tool: studying specific
combinations of invariants and operations can indicate the existence
of more scalable algorithms~\cite{kohler-commutativity}.

In summary, this paper offers the following high-level takeaways:
\begin{introenumerate}



\item Serializable transactions preserve application correctness at
  the cost of always coordinating between conflicting reads and writes.

\item Given knowledge of application transactions and correctness
  criteria (e.g., invariants), it is often possible to avoid this
  coordination (by executing some transactions without coordination,
  thus providing availability, low latency, and excellent scalability)
  while still preserving those correctness criteria.

\item Invariant confluence offers a necessary and sufficient
  condition for this correctness-preserving, coordination-free execution.

\item Many common integrity constraints found in SQL and standardized
  benchmarks are invariant confluent, allowing order-of-magnitude
  performance gains over coordinated execution.

\end{introenumerate}
While coordination cannot always be avoided, this work evidences the
power of application invariants in scalable and correct
execution of modern applications on modern hardware. Application
correctness does 
not always require coordination, and \iconfluence analysis can explain 
both when and why this is the case.

\ifextended
\minihead{Overview} The remainder of this paper proceeds as follows:
Section~\ref{sec:motivation} describes and quantifies the costs of
coordination. Section~\ref{sec:model} introduces our system model and
Section~\ref{sec:bcc-theory} contains our primary theoretical
result. Readers may skip to Section~\ref{sec:bcc-practice} for
practical applications of \iconfluence to real-world
invariant-operation combinations. Section~\ref{sec:evaluation}
subsequently applies
these combinations to real applications and presents an experimental
case study of TPC-C. Section~\ref{sec:relatedwork} describes related
work, while Section~\ref{sec:futurework} discusses possible extensions
and Section~\ref{sec:conclusion} concludes.
\else
\minihead{Overview} The remainder of this paper proceeds as follows:
Section~\ref{sec:motivation} describes and quantifies the costs of
coordination. Section~\ref{sec:model} introduces our system model and
Section~\ref{sec:bcc-theory} contains our primary theoretical
result. Readers may skip to Section~\ref{sec:bcc-practice} for
practical applications of \iconfluence to real-world
invariant-operation combinations. Section~\ref{sec:evaluation}
subsequently applies
these combinations to real applications and presents an experimental
case study of TPC-C. Section~\ref{sec:relatedwork} describes related
work, and Section~\ref{sec:conclusion} concludes.
\fi

\section{Conflicts and Coordination}
\label{sec:motivation}

As repositories for application state, databases are traditionally
tasked with maintaining correct data on behalf of users. During
concurrent access to data, a database ensuring correctness must
therefore decide which user operations can execute simultaneously and
which, if any, must coordinate, or block. In this section, we explore
the relationship between the correctness criteria that a database attempts
to maintain and the coordination costs of doing so.

\minihead{By example} As a running example, we consider a
database-backed payroll application that maintains information about
employees and departments within a small business. In the application,
$a.)$ each employee is assigned a unique ID number and $b.)$ each
employee belongs to exactly one department. A database ensuring correctness must
maintain these application-level properties, or \textit{invariants} on
behalf of the application (i.e., without application-level
intervention). In our payroll application, this is non-trivial: for
example, if the application attempts to simultaneously create two
employees, then the database must ensure the employees are
assigned distinct IDs.

\minihead{Serializability and conflicts} The classic answer to
maintaining application-level invariants is to use serializable
isolation: execute each user's ordered sequence of operations, or
\textit{transactions}, such that the end result is equivalent to some
sequential execution~\cite{tamer-book,bernstein-book,gray-virtues}. If
each transaction preserves correctness in isolation, composition via
serializable execution ensures correctness. In our payroll example,
the database would execute the two employee creation transactions such
that one transaction appears to execute after the other, avoiding 
duplicate ID assignment.

While serializability is a powerful abstraction, it comes with a cost:
for arbitrary transactions (and for all implementations of
serializability's more conservative variant---conflict
serializability), any two operations to the same item---at least one
of which is a write---will result in a \textit{read/write
  conflict}. Under serializability, these conflicts require
coordination or, informally, blocking communication between concurrent
transactions: to provide a serial ordering, conflicts must be totally
ordered across transactions~\cite{bernstein-book}. For example, given database state
$\{x=\bot, y=\bot\}$, if transaction $T_1$ writes $x=1$ and reads from
$y$ and $T_2$ writes $y=1$ and reads from $x$, a database cannot both
execute $T_1$ and $T_2$ entirely concurrently and maintain
serializability~\cite{davidson-survey,hat-vldb}.

\minihead{The costs of coordination} The coordination overheads above
incur three primary penalties: increased latency (due to stalled
execution), decreased throughput, and, in the event of partial
failures, unavailability. If a transaction takes $d$ seconds to
execute, the maximum throughput of conflicting transactions operating
on the same items under a general-purpose (i.e., interactive,
non-batched) transaction model is limited by $\frac{1}{d}$, while
coordinating operations will also have to wait. On a single system,
delays can be small, permitting tens to hundreds of thousands of
conflicting transactions per item per second. In a partitioned
database system, where different items are located on different
servers, or in a replicated database system, where the same item is
located (and is available for operations) on multiple servers, the
cost increases: delay is lower-bounded by network latency. On a local
area network, delay may vary from several microseconds (e.g., via
Infiniband or RDMA) to several milliseconds on today's cloud
infrastructure, permitting anywhere from a few hundred transactions to
a few hundred thousand transactions per second. On a wide-area
network, delay is lower-bounded by the speed of light (worst-case on
Earth, around $75$ms, or about 13 operations per
second~\cite{hat-vldb}). Under network
partitions~\cite{queue-partitions}, as delay tends towards infinity,
these penalties lead to unavailability~\cite{gilbert-cap,hat-vldb}. In
contrast, operations executing without coordination can proceed
concurrently and will not incur these penalties.

\minihead{Quantifying coordination overheads} To further understand the
costs of coordination, we performed two sets of measurements---one using
a database prototype and one using traces from
prior studies.

\begin{figure}
\includegraphics[width=\columnwidth]{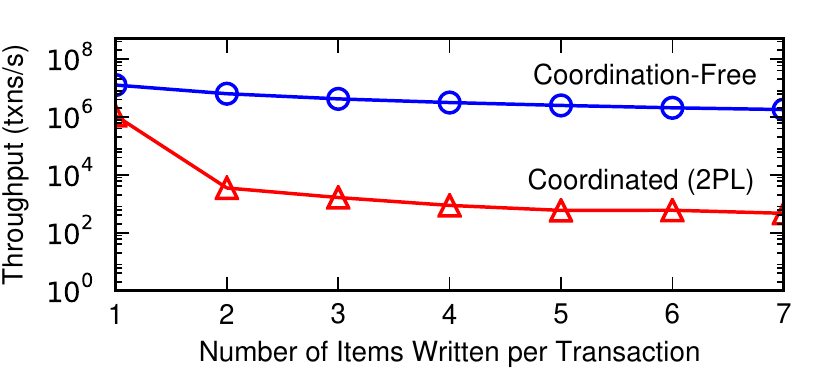}\vspace{-.5em}
\caption{Microbenchmark performance of coordinated and
  coordination-free execution of transactions of varying size writing
  to eight items located on eight separate multi-core servers.}
\label{fig:micro}
\end{figure}

We first compared the throughput of a set of coordinated and
coordination-free transaction execution. We partitioned a set
of eight data items across eight servers and ran one set of
transactions with an optimized variant of two-phase locking (providing
serializability)~\cite{bernstein-book} and ran another set of transactions
without coordination (Figure~\ref{fig:micro}; see
\rappendix{\appexperiment} for more details). With
single-item, non-distributed transactions, the coordination-free
implementation achieves, in aggregate, over 12M transactions per
second and bottlenecks on \textit{physical resources}---namely, CPU
cycles. In contrast, the lock-based implementation achieves
approximately $1.1$M transactions per second: it is unable to fully
utilize all multi-core processor contexts due to lock contention. For
distributed transactions, coordination-free throughput decreases linearly
(as an $N$-item transaction performs $N$ writes), while the throughput
of coordinating transactions drops by over three orders of
magnitude.

While the above microbenchmark demonstrates the costs of a particular
\textit{implementation} of coordination, we also studied the
effect of more fundamental, implementation-independent overheads
(i.e., also applicable to optimistic and scheduling-based concurrency
control mechanisms). We determined the maximum attainable throughput
for coordinated execution within a single datacenter (based on data
from~\cite{bobtail}) and across multiple datacenters (based on data
from~\cite{hat-vldb}) due to blocking coordination during atomic
commitment~\cite{bernstein-book}. For an $N$-server transaction,
classic two-phase commit (\cpc) requires $N$ (parallel) coordinator to
server RTTs, while decentralized two-phase commit (\dpc) requires $N$
(parallel) server to server broadcasts, or $N^2$
messages. Figure~\ref{fig:2pc} shows that, in the local area, with
only two servers (e.g., two replicas or two coordinating operations on
items residing on different servers), throughput is bounded by $1125$
transactions/s (via \dpc; $668$/s via \cpc). Across eight servers,
\dpc throughput drops to $173$ transactions/s (resp. $321$ for \cpc)
due to long-tailed latency distributions. In the wide area, the
effects are more stark: if coordinating from Virginia to Oregon, \dpc
message delays are $83$~ms per commit, allowing $12$ operations per
second. If coordinating between all eight EC2 availability zones,
throughput drops to slightly over $2$ transactions/s in both
algorithms. (\rappendix{\appexperiment} provides more details.)

These results should be unsurprising: coordinating---especially over
the network---can incur serious performance penalties. In contrast,
coordination-free operations can execute without incurring these
costs. The costs of actual workloads can vary: if coordinating
operations are rare, concurrency control will not be a bottleneck. For
example, a serializable database executing transactions with disjoint
read and write sets can perform as well as a non-serializable database
without compromising correctness~\cite{shore-communication}. However,
as these results demonstrate, minimizing the amount of coordination
and its degree of distribution can therefore have a tangible impact on
performance, latency, and
availability~\cite{pacelc,hat-vldb,gilbert-cap}. While we study real
applications in Section~\ref{sec:evaluation}, these measurements
highlight the worst of coordination costs on modern hardware.

\begin{figure}
  \includegraphics[width=\columnwidth]{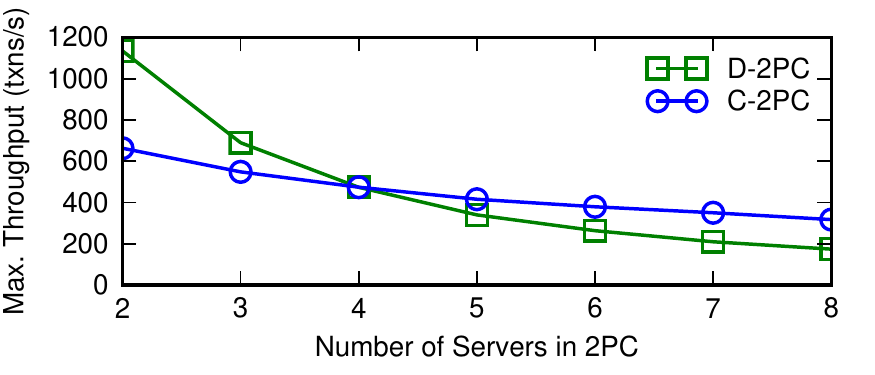}\\
  {\centering \textbf{\scriptsize a.) Maximum transaction
      throughput over local-area network in~\cite{bobtail}}\par}
  \includegraphics[width=\columnwidth]{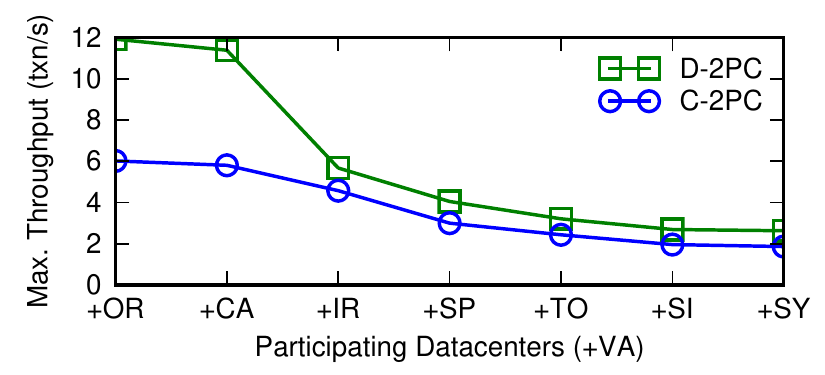}\\
  \textbf{\scriptsize b.) Maximum throughput over wide-area network
    in~\cite{hat-vldb} with transactions originating from a
    coordinator in Virginia (VA; OR:~Oregon, CA:~California,
    IR:~Ireland, SP:~S\~{a}o Paulo, TO:~Tokyo, SI:~Singapore,
    SY:~Sydney)}

\caption{Atomic commitment latency as an upper bound on throughput
  over LAN and WAN networks.}
\label{fig:2pc}
\end{figure}

\minihead{Our goal: Minimize coordination} In this paper, we seek to
minimize the amount of coordination required to correctly execute an
application's transactions. As discussed in Section~\ref{sec:intro},
serializability is \textit{sufficient} to maintain correctness but is
not always \textit{necessary}; that is, many---but not
all---transactions can be executed concurrently without necessarily
compromising application correctness. In the remainder of this paper,
we identify when safe, coordination-free execution is possible. If
serializability requires coordinating between each possible pair of
conflicting reads and writes, we will only coordinate between pairs of
operations that might compromise \textit{application-level}
correctness. To do so, we must both raise the specification of
correctness beyond the level of reads and writes and directly account
for the process of reconciling the effects of concurrent transaction
execution at the application level.

\section{System Model}
\label{sec:model}

To characterize coordination avoidance, we first present a system
model. We begin with an informal overview. In our model, transactions
operate over independent (logical) ``snapshots'' of database
state. Transaction writes are applied at one or more snapshots
initially when the transaction commits and then are integrated into
other snapshots asynchronously via a ``merge'' operator that
incorporates those changes into the snapshot's state. Given a set of
invariants describing valid database states, as
Table~\ref{table:requirements} outlines, we seek to understand when it
is possible to ensure invariants are always satisfied (global
validity) while guaranteeing a response (transactional availability)
and the existence of a common state (convergence), all without
communication during transaction execution
(coordination-freedom). This model need not directly correspond
to a given implementation (e.g., see the database architecture in
Section~\ref{sec:evaluation})---rather, it serves as a useful
abstraction.  The remainder of this section
further defines these concepts; readers more interested in their
application should proceed to Section~\ref{sec:bcc-theory}. We provide
greater detail and additional discussion in~\cite{ca-extended}.

\begin{table}[t]
\begin{center}
\small
\begin{tabular}{|l|r|}
  \hline\textbf{Property} & \textbf{Effect}  \\\hline
  Global validity & Invariants hold over committed states  \\
  Transactional availability & Non-trivial response guaranteed \\
  Convergence & Updates are reflected in shared state \\
  Coordination-freedom & No synchronous coordination\\\hline
\end{tabular}
\end{center}\vspace{-1em}
\caption{Key properties of the system model and their effects.}
\label{table:requirements}
\end{table}

\minihead{Databases} We represent a state of the shared database as a
set $D$ of unique \textit{versions} of data items located on an arbitrary set of
database servers, and each version is located on at least one
server. We use ${\cal D}$ to denote the set of possible database
states---that is, the set of sets of versions. The database is
initially populated by an initial state $D_0$ (typically but not
necessarily empty).

\minihead{Transactions, Replicas, and Merging} Application clients
submit requests to the database in the form of transactions, or
ordered groups of operations on data items that should be executed
together. Each transaction operates on a logical \textit{replica}, or
set of versions of the items mentioned in the transaction. At the
beginning of the transaction, the replica contains a subset of the
database state and is formed from all of the versions of the relevant
items that can be found at one or more physical servers that are
contacted during transaction execution. As the transaction executes,
it may add versions (of items in its writeset) to its replica. Thus,
we define a transaction $T$ as a transformation on a replica: $T:
{\cal D}\rightarrow {\cal D}$. We treat transactions as opaque
transformations that can contain writes (which add new versions to the
replica's set of versions) or reads (which return a specific set of
versions from the replica). (Later, we will discuss transactions
operating on data types such as counters.)

Upon completion, each transaction can \textit{commit}, signaling
success, or \textit{abort}, signaling failure. Upon commit, the
replica state is subsequently \textit{merged} ($\sqcup$:${\cal D}
\times {\cal D} \rightarrow {\cal D}$) into the set of versions at
least one server. We require that the merged effects of a committed
transaction will eventually become visible to other transactions that
later begin execution on the same server.\footnote{This implicitly
  disallows servers from always returning the initial database state
  when they have newer writes on hand. This is a relatively pragmatic
  assumption but also simplifies our later reasoning about admissible
  executions. This assumption could possibly be relaxed by adapting
  Newman's lemma~\cite{obs-confluence}, but we do not consider the
  possibility here.}  Over time, effects propagate to other servers,
again through the use of the merge operator.  Though not strictly
necessary, we assume this merge operator is commutative, associative,
and idempotent~\cite{calm,crdt} and that, for all states $D_i$, $D_0
\sqcup D_i = D_i$.  In our initial model, we define
merge as set union of the versions contained at different servers.
(Section~\ref{sec:bcc-practice} discusses additional implementations.)
For example, if server $R_x = \{v\}$ and $R_y = \{w\}$, then $R_x
\sqcup R_y = \{v, w\}$.

In effect, each transaction can modify its replica state without
modifying any other concurrently executing transactions' replica
state. Replicas therefore provide transactions with partial
``snapshot'' views of global state (that we will use to simulate
concurrent executions, similar to revision
diagrams~\cite{ec-txns}). Importantly, two transactions' replicas do
not necessarily correspond to two physically separate servers; rather,
a replica is simply a partial ``view'' over the global state of the
database system. For now, we assume transactions are known in advance
(see also~\rsection{8}).

\minihead{Invariants} To determine whether a database state is valid
according to application correctness criteria, we use
\textit{invariants}, or predicates over replica state: $I: {\cal D}
\rightarrow \{true, false\}$~\cite{eswaran-consistency}.  In our
payroll example, we could specify an invariant that only one user in a
database has a given ID. This invariant---as well as almost all
invariants we consider---is naturally expressed as a part of the
database schema (e.g., via DDL); however, our approach allows us to
reason about invariants even if they are known to the developer but
not declared to the system. Invariants directly capture the notion of
ACID Consistency~\cite{bernstein-book,gray-virtues}, and we say that a
database state is \textit{valid} under an invariant $I$ (or $I$-valid)
if it satisfies the predicate:

\begin{definition}
A replica state $R \in {\cal D}$ is \textit{$I$-valid} iff $I(R) = true$.
\end{definition}

We require that $D_0$ be valid under
invariants. Section~\ref{sec:theory-discussion} provides additional discussion
regarding our use of invariants.


\minihead{Availability} To ensure each transaction receives a
non-trivial response,
we adopt the following definition of
\textit{availability}~\cite{hat-vldb}:

\begin{definition} 
  A system provides \textit{transactionally available} execution iff,
  whenever a client executing a transaction $T$ can access servers
  containing one or more versions of each item in $T$, then $T$
  eventually commits or aborts itself either due to an \textit{abort}
  operation in $T$ or if committing the transaction would violate a
  declared invariant over $T$'s replica state. $T$ will commit in all
  other cases.
\end{definition}

Under the above definition, a transaction can only abort if it
explicitly chooses to abort itself or if committing would violate
invariants over the transaction's replica state.\footnote{This basic
  definition precludes fault tolerance (i.e., durability) guarantees
  beyond a single server failure~\cite{hat-vldb}. We can relax this
  requirement and allow communication with a fixed number of servers
  (e.g., $F+1$ servers for $F$-fault tolerance; $F$ is often
  small~\cite{dynamo}) without affecting our results. This does not
  affect scalability because, as more replicas are added, the
  communication overhead required for durability remains constant.}

\begin{figure}
\begin{center}
\includegraphics[width=.85\columnwidth]{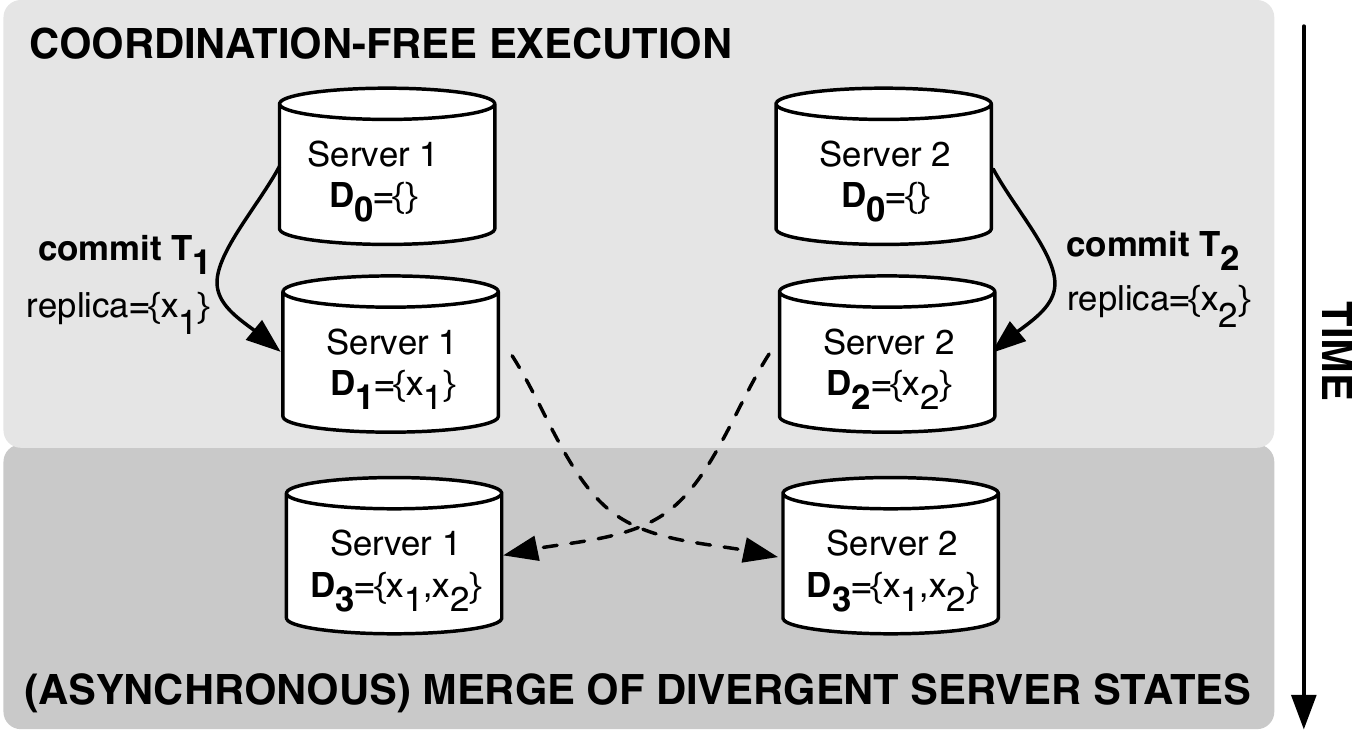}
\end{center}\vspace{-2em}
\caption{An example coordination-free execution of two transactions,
  $T_1$ and $T_2$, on two servers. Each transaction writes to its
  local replica, then, after commit, the servers asynchronously
  exchange state and converge to a common state ($D_3$).}
\label{fig:replicas}
\end{figure}

\minihead{Convergence} Transactional availability allows replicas to
maintain valid state \textit{independently}, but it is vacuously
possible to maintain ``consistent'' database states by letting
replicas diverge (contain different state) forever. This guarantees
\textit{safety} (nothing bad happens) but not \textit{liveness}
(something good happens)~\cite{schneider-concurrent}. To enforce state
sharing, we adopt the following definition:

\begin{definition}
  A system is \textit{convergent} iff, for each pair of servers, in
  the absence of new writes to the servers and in the
  absence of indefinite communication delays between the servers,
  the servers eventually contain the same versions for any item
  they both store.
\end{definition}

To capture the process of reconciling divergent states, we use the
previously introduced merge operator: given two divergent server
states, we apply the merge operator to produce convergent state. We
assume the effects of merge are atomically visible: either all effects
of a merge are visible or none are. This assumption is not always
necessary but it simplifies our discussion and, as we later discuss,
is maintainable without coordination~\cite{ramp-txns,hat-vldb}.

\minihead{Maintaining validity} To make sure that both divergent and
convergent database states are valid and, therefore, that transactions
never observe invalid states, we introduce the following property:

\begin{definition}
A system is \textit{globally $I$-valid} iff all replicas always contain
$I$-valid state.
\end{definition}

\minihead{Coordination} Our system model is missing one final
constraint on coordination between concurrent transaction
execution:

\begin{definition}
  A system provides coordination-free execution for a set of
  transactions $T$ iff the progress of executing each $t\in T$ is only
  dependent on the versions of the items $t$ reads (i.e., $t$'s replica
  state).
\end{definition}

\noindent That is, in a coordination-free execution, each
transaction's progress towards commit/abort is independent of other
operations (e.g., writes, locking, validations) being performed on
behalf of other transactions. This precludes blocking
synchronization or communication across concurrently executing transactions.

\minihead{By example} Figure~\ref{fig:replicas} illustrates a
coordination-free execution of two transactions $T_1$ and $T_2$ on two
separate, fully-replicated physical servers. Each transaction commits
on its local replica, and the result of each transaction is reflected
in the transaction's local server state. After the transactions have
completed, the servers exchange state and, after applying the merge
operator, converge to the same state. Any transactions executing later on
either server will obtain a replica that includes the effects of both transactions.

\ifextended
\subsection{Extended Notes}

Our treatment of convergence uses a pair-wise definition (i.e., each
pair converges) rather than a system-wide definition (i.e., all nodes
converge). This is more restrictive than system-wide convergence but
allows us to make guarantees on progress despite partitions
between subsets of the servers (notably, precludes the use of protocols such
as background consensus, which can stall indefinitely in the presence
of partitions). Like many of the other decisions in our model, this
too could likely be relaxed at the cost of a less friendly
presentation of the concepts below.

The above decisions (including those footnoted in the primary text)
were made to strike a balance between generality and ease of
presentation in later applications of these concepts. In practice,
every non-\iconfluent set of transactions and invariants we
encountered (see
Sections~\ref{sec:bcc-theory},~\ref{sec:bcc-practice}) had a
counter-example consisting of a divergent execution consisting of a
single pair of transactions. However, we admit the possiblity that
more exotic transactions and merge functions might require in more
complicated histories, so we consider arbitrary histories
below. Precisely characterizing the expressive power of executions (in
terms of admissible output states) under a single-transaction
divergence versus multi-transaction divergence is a fairly interesting
question for future work.

\fi

\section{Consistency sans Coordination}
\label{sec:bcc-theory}

With a system model and goals in hand, we now address the question:
when do applications require coordination for correctness? The answer
depends not just on an application's transactions or on an 
application's invariants. Rather, the answer depends on the
\textit{combination} of the two under consideration. Our contribution
in this section is to formulate a criterion that will answer this
question for specific combinations in an implementation-agnostic
manner.

In this section, we focus almost exclusively on providing a formal
answer to this question. The remaining sections of this paper are
devoted to practical interpretation and application of these results.

\subsection{\iconfluence: Criteria Defined}

To begin, we introduce the central property (adapted from the constraint
programming literature~\cite{obs-confluence}) in our main
result: invariant confluence (hereafter, \iconfluence). Applied in a
transactional context, the \iconfluence property informally ensures
that divergent but $I$-valid database states can be merged into a
valid database state---that is, the set of valid states reachable by
executing transactions and merging their results is closed
(w.r.t. validity) under merge. In the next sub-section, we show that
\iconfluence analysis directly determines the potential for safe,
\cfree execution.

We say that $S_i$ is a \textit{$I$-$T$-reachable state} if, given an
invariant $I$ and set of transactions $T$ (with merge function
$\sqcup$), there exists a (partially ordered) sequence of transaction
and merge function invocations that yields $S_i$, and each
intermediate state produced by transaction execution or merge
invocation is also $I$-valid. We call these previous states
\textit{ancestor states}. Note that each ancestor state is either
$I$-$T$-reachable or is instead the initial state ($D_0$).

We can now formalize the \iconfluence property:

\begin{definition}[\iconfluence]
  A set of transactions $T$ is \iconfluent with respect to invariant
  $I$ if, for all $I$-$T$-reachable states $D_i$, $D_j$ with a common
  ancestor state, $D_i \sqcup D_j$ is $I$-valid.
\end{definition}

Figure~\ref{fig:iconfluence} depicts an \iconfluent merge of
two $I$-$T$-reachable states, each starting from a shared,
$I$-$T$-reachable state $D_s$. Two sequences of transactions
$t_{in}\dots t_{i1}$ and $t_{jm}\dots t_{j1}$ each independently
modify $D_s$. Under \iconfluence, the states produced by these
sequences ($D_{in}$ and $D_{jm}$) must be valid under
merge.\footnote{We require these states to have a common
  ancestor to rule out the possibility of merging states that could
  not have arisen from transaction execution (e.g., even if no
  transaction assigns IDs, merging two states that each have unique
  but overlapping sets of IDs could be invalid).}

\iconfluence holds for specific combinations of invariants and
transactions. In our payroll database example from
Section~\ref{sec:motivation}, removing a user from the database is
\iconfluent with respect to the invariant that user IDs are
unique. However, two transactions that remove two different users from
the database are not \iconfluent with respect to the invariant that
there exists at least one user in the database at all
times. Section~\ref{sec:bcc-practice} discusses additional
combinations of invariants (with greater precision).

\begin{figure}
\begin{center}
\includegraphics[width=\columnwidth]{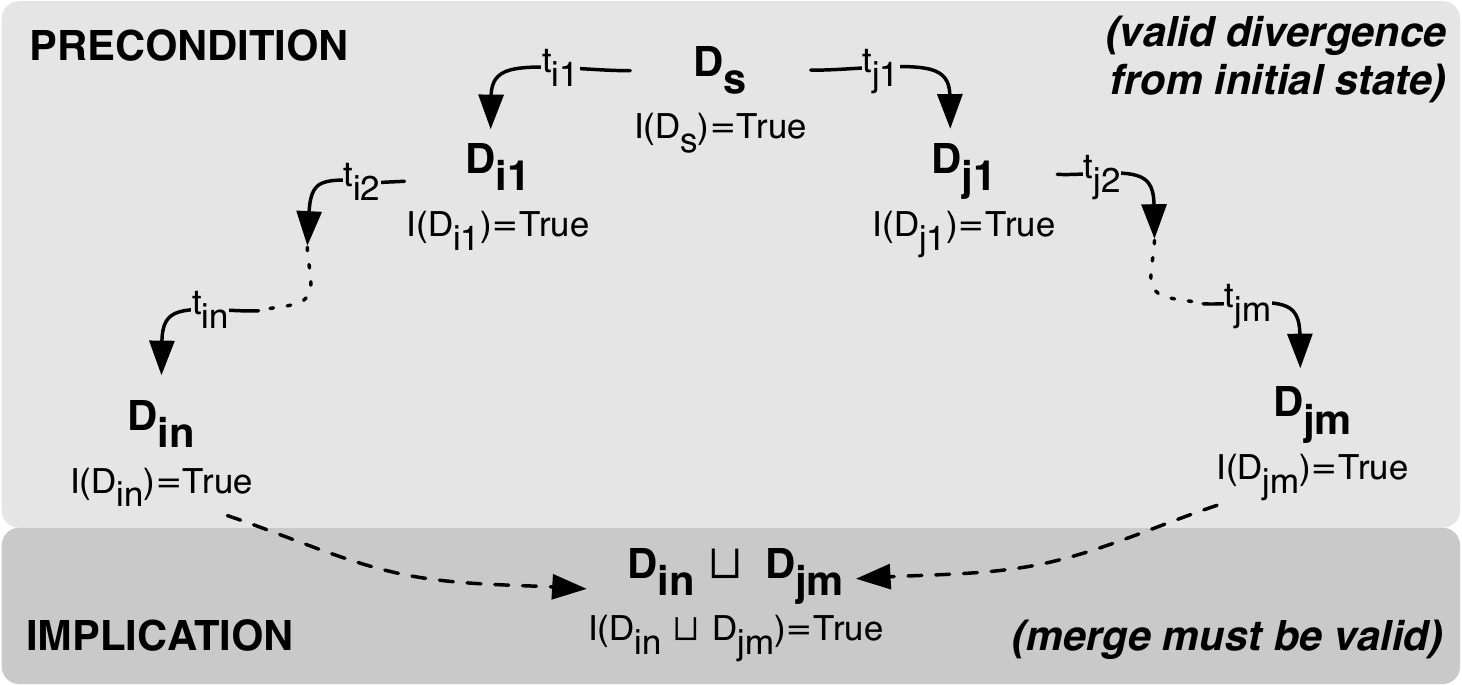}\vspace{-1em}
\end{center}
\caption{An \iconfluent execution illustrated via a diamond
  diagram. If a set of transactions $T$ is \iconfluent, then all
  database states reachable by executing and merging transactions in
  $T$ starting with a common ancestor ($D_s$) must
  be mergeable ($\sqcup$) into an $I$-valid database state.}
\label{fig:iconfluence}
\end{figure}

\subsection{\iconfluence and Coordination}
\label{sec:ic-result}

We can now apply \iconfluence to our goals from Section~\ref{sec:model}:

\begin{theorem}
\label{theorem:necessary}
A globally $I$-valid system can execute a set of transactions $T$ with
\cfreedom, transactional availability, convergence if and only if $T$
is \iconfluent with respect to $I$.
\end{theorem}

We provide a full proof of Theorem~\ref{theorem:necessary}
in~\rappendix{\apptheory} (which is straightforward) but provide a
sketch here. The backwards direction is by construction: if
\iconfluence holds, each replica can check each transaction's
modifications locally and replicas can merge independent modifications
to guarantee convergence to a valid state. The forwards direction uses
a partitioning argument~\cite{gilbert-cap} to derive a contradiction:
we construct a scenario under which a system cannot determine whether
a non-\iconfluent transaction should commit without violating one of
our desired properties (either compromising validity or availability,
diverging forever, or coordinating).

Theorem~\ref{theorem:necessary} establishes \iconfluence as a
necessary and sufficient condition for invariant-preserving,
coordination-free execution.  If \iconfluence holds, there exists a
correct, coordination-free execution strategy for the transactions; if
not, no possible implementation can guarantee these properties for the
provided invariants and transactions. That is, if \iconfluence does
not hold, there exists at least one execution of transactions on
separate replicas that will violate the given invariants when servers
converge. To prevent invalid states from occurring, at least one of
the transaction sequences will have to forego availability or
\cfreedom, or the system will have to forego convergence. \iconfluence
analysis is independent of any given implementation, and effectively
``lifts'' prior discussions of scalability, availability, and low
latency~\cite{hat-vldb,gilbert-cap,pacelc} to the level of application
(i.e., not ``I/O''~\cite{consistency-borders}) correctness. This
provides a useful handle on the implications of coordination-free
execution without requiring reasoning about low-level properties such
as physical data location and the number of servers.

\subsection{Discussion and Limitations}
\label{sec:theory-discussion}

\iconfluence captures a simple (informal) rule:
\textbf{\textit{coordination can only be avoided if all local commit
    decisions are globally valid}}. (Alternatively, commit decisions
are composable.) If two independent decisions to commit can result in
invalid converged state, then replicas must coordinate in order to
ensure that only one of the decisions is to commit. Given the
existence of an unsafe execution and the inability to distinguish
between safe and invalid executions using only local information, a
globally valid system \textit{must} coordinate in order to prevent the
invalid execution from arising.

\minihead{Use of invariants} Our use of invariants in \iconfluence is
key to achieving a \textit{necessary} and not simply sufficient
condition. By directly capturing application-level correctness
criteria via invariants, \iconfluence analysis only identifies
``true'' conflicts. This allows \iconfluence analysis to perform a
more accurate assessment of whether coordination is needed compared to
related conditions such as commutativity
(Section~\ref{sec:relatedwork}).

However, the reliance on invariants also has drawbacks. \iconfluence
analysis only guards against violations of any provided invariants. If
invariants are incorrectly or incompletely specified, an \iconfluent
database system may violate application-level correctness. If users
cannot guarantee the correctness and completeness of their invariants
and operations,
they should opt for a more conservative analysis or mechanism such as
employing serializable transactions. Accordingly, our development of
\iconfluence analysis provides developers with a powerful option---but
only if used correctly. If used incorrectly, \iconfluence allows
incorrect results, or, if not used at all, developers must resort to existing alternatives.

This final point raises several questions: can we specify invariants
in real-world use cases? Classic database concurrency control models
assume that ``the [set of application invariants] is generally not
known to the system but is embodied in the structure of the
transaction''~\cite{traiger-tods,eswaran-consistency}. Nevertheless,
since 1976, databases have introduced support for a finite set of
invariants~\cite{korth-serializability,decomp-semantics,garciamolina-semantics,ic-survey,ic-survey-two}
in the form of primary key, foreign key, uniqueness, and row-level
``check'' constraints~\cite{kemme-si-ic}. We can (and, in this paper,
do) analyze these invariants, which can---like many program
analyses~\cite{kohler-commutativity}---lead to new insights about
execution strategies. We have found the process of invariant
specification to be non-trivial but feasible in practice;
Section~\ref{sec:evaluation} describes some of our experiences.

\minihead{(Non-)determinism} \iconfluence analysis effectively
captures points of \textit{unsafe
  non-determinism}~\cite{consistency-borders} in transaction
execution. As we have seen in many of our examples thus far, total
non-determinism under concurrent execution can compromise
application-level consistency~\cite{calm,termrewriting}. But not all
non-determinism is bad: many desirable properties (e.g.,
classical distributed consensus among processes) involve forms 
of acceptable non-determinism (e.g., any proposed outcome is
acceptable as long as all processes agree)~\cite{paxos-commit}. In
many cases, maximizing safe concurrency requires non-determinism.

\iconfluence analysis allows this non-deterministic divergence of
database states but makes two useful guarantees about those
states. First, the requirement for global validity ensures safety (in
the form of invariants). Second, the requirement for convergence
ensures liveness (in the form of convergence). Accordingly, via its
use of invariants, \iconfluence allows users to scope non-determinism
while permitting only those states that are acceptable.

\section{Applying Invariant Confluence}
\label{sec:bcc-practice}
\label{sec:merge}

As a test for coordination requirements, \iconfluence exposes a
trade-off between the operations a user wishes to perform and the
properties she wishes to guarantee. At one extreme, if a user's
transactions do not modify database state, she can guarantee any satisfiable
invariant. At the other extreme, with no invariants, a user can safely
perform any operations she likes. The space in-between contains a spectrum of
interesting and useful combinations.

Until now, we have been largely concerned with formalizing
\iconfluence for abstract operations; in this section, we begin to
leverage this property. We examine a series of practical
invariants by considering several features of SQL, ending with
abstract data types and revisiting our payroll example along the
way. We will apply these results to full applications in
Section~\ref{sec:evaluation}.

In this section, we focus on providing intuition and informal
explanations of our \iconfluence analysis. Interested readers can find
a more formal analysis in \rappendix{\appapply},
including discussion of invariants not presented here. For convenience,
we reference specific proofs from \rappendix{\appapply} inline.

\begin{table}
\definecolor{yesgray}{gray}{0.92}
\begin{center}
\small
\begin{tabular}{|l|l|c|c|}
\hline
\textbf{Invariant} & \textbf{Operation} & $\mathcal{I}$-C? & Proof \#\\\hline

\rowcolor{yesgray}
Attribute Equality & Any & Yes &1\\
\rowcolor{yesgray}
Attribute Inequality & Any & Yes&2 \\
Uniqueness & Choose specific value & No&3\\
\rowcolor{yesgray}
Uniqueness & Choose some value & Yes&4\\
\texttt{AUTO\_INCREMENT} & Insert & No&5\\
\rowcolor{yesgray}
Foreign Key & Insert & Yes&6\\
Foreign Key & Delete & No&7\\
\rowcolor{yesgray}
Foreign Key & Cascading Delete & Yes&8\\
\rowcolor{yesgray}
Secondary Indexing & Update & Yes &9\\
\rowcolor{yesgray}
Materialized Views & Update & Yes &10\\\hline\hline
\rowcolor{yesgray}
> & Increment [Counter] & Yes &11 \\
< & Increment [Counter] & No &12\\
> & Decrement [Counter] & No &13\\
\rowcolor{yesgray}
< & Decrement [Counter] & Yes &14\\
\rowcolor{yesgray}
\texttt{[NOT] CONTAINS} & Any [Set, List, Map] & Yes &15, 16\\ 
\texttt{SIZE=} & Mutation [Set, List, Map] & No &17\\ \hline
\end{tabular}
\end{center}
\caption{Example SQL (top) and ADT invariant \iconfluence along with
  references to formal proofs in \rappendix{\appapply}.}
\label{table:invariants}
\end{table}

\subsection{\iconfluence for Relations}

We begin by considering several constraints found in SQL.

\minihead{Equality} As a warm-up, what if an application wants to
prevent a particular value from appearing in a database? For example,
our payroll application from Section~\ref{sec:motivation} might
require that every user have a last name, marking the \texttt{LNAME}
column with a \texttt{NOT NULL} constraint. While not particularly
exciting, we can apply \iconfluence analysis to insertions and updates
of databases with (in-)equality constraints (Claims 1, 2 in
\rappendix{\appapply}). Per-record inequality invariants are
\iconfluent, which we can show by contradiction: assume two database
states $S_1$ and $S_2$ are each $I$-$T$-reachable under per-record
in-equality invariant $I_e$ but that $I_e(S_1 \sqcup S_2)$ is
false. Then there must be a $r \in S_1 \sqcup S_2$ that violates $I_e$
(i.e., $r$ has the forbidden value). $r$ must appear in $S_1$, $S_2$,
or both. But, that would imply that one of $S_1$ or $S_2$ is not
$I$-valid under $I_e$, a contradiction.

\minihead{Uniqueness} We can also consider common uniqueness
invariants (e.g., \texttt{PRIMARY KEY} and \texttt{UNIQUE}
constraints). For example, in our payroll example, we wanted user IDs
to be unique. In fact, our earlier discussion in
Section~\ref{sec:motivation} already provided a counterexample showing
that arbitrary insertion of users is not \iconfluent under these
invariants: $\{$Stan:5$\}$ and $\{$Mary:5$\}$ are both
$I$-$T$-reachable states that can be created by a sequence of
insertions (starting at $S_0=\{\}$), but their merge---$\{$Stan:5,
Mary:5$\}$---is not $I$-valid. Therefore, uniqueness is not
\iconfluent for inserts of unique values (Claim 3). However, reads and
deletions are both \iconfluent under uniqueness invariants: reading
and removing items cannot introduce duplicates.

Can the database safely \textit{choose} unique values on behalf of
users (e.g., assign a new user an ID)? In this case, we can achieve
uniqueness without coordination---as long as we have a notion of
replica membership (e.g., server or replica IDs). The difference is
subtle (``grant this record this specific, unique ID'' versus ``grant
this record some unique ID''), but, in a system model
with membership (as is practical in many contexts), is powerful. If
replicas assign unique IDs within their respective portion of the ID
namespace, then merging locally valid states will also be globally
valid (Claim 4).

\minihead{Foreign Keys} We can consider more complex invariants, such
as foreign key constraints. In our payroll example, each employee
belongs to a department, so the application could specify a constraint
via a schema declaration to capture this relationship (e.g.,
\texttt{EMP.D\_ID FOREIGN KEY REFERENCES DEPT.ID}).

Are foreign key constraints maintainable without coordination? Again,
the answer depends on the actions of transactions modifying the data
governed by the invariant. Insertions under foreign key constraints
\textit{are} \iconfluent (Claim 6). To show this, we again attempt to find two
$I$-$T$-reachable states that, when merged, result in invalid
state. Under foreign key constraints, an invalid state will contain a
record with a ``dangling pointer''---a record missing a corresponding
record on the opposite side of the association. If we assume there
exists some invalid state $S_1 \sqcup S_2$ containing a record $r$
with an invalid foreign key to record $f$, but $S_1$ and $S_2$ are
both valid, then $r$ must appear in $S_1$, $S_2$, or both. But, since
$S_1$ and $S_2$ are both valid, $r$ must have a corresponding foreign
key record ($f$) that ``disappeared'' during merge. Merge (in the
current model) does not remove versions, so this is impossible.

From the perspective of \iconfluence analysis, foreign key constraints
concern the \textit{visibility} of related updates: if individual
database states maintain referential integrity, a non-destructive
merge function such as set union cannot cause tuples to ``disappear''
and compromise the constraint. This also explains why models such as
read committed~\cite{adya-isolation} and read
atomic~\cite{adya-isolation} isolation as well as causal
consistency~\cite{hat-vldb} are also achievable without coordination:
simply restricting the visibility of updates in a given transaction's
read set does not require coordination between concurrent operations.

Deletions and modifications under foreign key constraints are more
challenging. Arbitrary deletion of records is unsafe: a user might be
added to a department that was concurrently deleted (Claim
7). However, performing cascading deletions (e.g., SQL \texttt{DELETE
  CASCADE}), where the deletion of a record also deletes \textit{all}
matching records on the opposite end of the association, is
\iconfluent under foreign key constraints (Claim 8). We can generalize
this discussion to updates (and cascading updates).

\minihead{Materialized Views} Applications often pre-compute results
to speed query performance via a materialized view~\cite{tamer-book}
(e.g., \texttt{UNREAD\_CNT} as \texttt{SELECT}\texttt{
}\texttt{COUNT(*)}\texttt{ }\texttt{FROM}\texttt{
}\texttt{emails}\texttt{ }\texttt{WHERE}\texttt{ }\texttt{read\_date =
  NULL}). We can consider a class of invariants that specify that
materialized views reflect primary data; when a transaction (or merge
invocation) modifies data, any relevant materialized views should be
updated as well. This requires installing updates at the same time as
the changes to the primary data are installed (a problem related to
maintaining foreign key constraints). However, given that a view
only reflects primary data, there are no ``conflicts.'' Thus,
materialized view maintenance updates are \iconfluent (Claim 10).

\subsection{\iconfluence for Data Types}

So far, we have considered databases that store growing sets of
immutable versions. We have used this model to analyze several useful
constraints, but, in practice, databases do not (often) provide these
semantics, leading to a variety of interesting anomalies. For example,
if we implement a user's account balance using a ``last writer wins''
merge policy~\cite{crdt}, then performing two concurrent withdrawal
transactions might result in a database state reflecting only one
transaction (a classic example of the Lost Update
anomaly)~\cite{adya-isolation,hat-vldb}. To avoid variants of these
anomalies, many optimistic, coordination-free database designs have
proposed the use of \textit{abstract data types} (ADTs), providing
merge functions for a variety of uses such as counters, sets, and
maps~\cite{crdt,atomictransactions,weihl-thesis,blooml} that ensure
that all updates are reflected in final database state. For example, a
database can represent a simple counter ADT by recording the number of
times each transaction performs an \texttt{increment} operation on the
counter~\cite{crdt}.

\iconfluence analysis is also applicable to these ADTs and their
associated invariants. For example, a row-level ``greater-than''
(\texttt{>}) threshold invariant is \iconfluent for counter
\texttt{increment} and \texttt{assign} ($\gets$) but not
\texttt{decrement} (Claims 11, 13), while a row-level ``less-than''
(\texttt{<}) threshold invariant is \iconfluent for counter
\texttt{decrement} and \texttt{assign} but not \texttt{increment}
(Claims 12, 14). This means that, in our payroll example, we can
provide \cfree support for concurrent salary increments but not
concurrent salary decrements. ADTs (including lists, sets, and maps)
can be combined with standard relational constraints like materialized
view maintenance (e.g., the ``total salary'' row should contain the
sum of employee salaries in the \texttt{employee} table).  This
analysis presumes user program explicitly use ADTs, and, as with our
generic set-union merge, \iconfluence ADT analysis requires a
specification of the ADT merge behavior (\rappendix{\appapply}
provides several examples).

\subsection{Discussion and Limitations}

We have analyzed a number of combinations of invariants and operations
(shown in Table~\ref{table:invariants}). These results are by no
means comprehensive, but they are expressive for many applications
(Section~\ref{sec:evaluation}). In this section, we discuss lessons from this
classification process.

\minihead{Analysis mechanisms} Here (and in~\rappendix{\appapply}), we
manually analyzed particular invariant and operation combinations,
demonstrating each to be \iconfluent or not. To study actual
applications, we can apply these labels via simple static
analysis. Specifically, given invariants (e.g., captured via SQL DDL)
and transactions (e.g., expressed as stored procedures), we can
examine each invariant and each operation within each transaction and
identify pairs that we have labeled as \iconfluent or
non-\iconfluent. Any pairs labeled as \iconfluent can be marked as
safe, while, for soundness (but not completeness), any unrecognized
operations or invariants can be flagged as potentially
non-\iconfluent. Despite its simplicity (both conceptually and in
terms of implementation), this technique---coupled with the results of
Table~\ref{table:invariants}---is sufficiently powerful to
automatically characterize the I-confluence of the applications we
consider in Section~\ref{sec:evaluation} when expressed in SQL (with
support for multi-row aggregates like Invariant 8 in
Table~\ref{table:tpcc-invariants}).

By growing our recognized list of \iconfluent pairs on an as-needed
basis (via manual analysis of the pair), the above technique has
proven useful---due in large part to the common re-use of invariants
like foreign key constraints. However, one could use more complex
forms of program analysis. For example, one might analyze the
\iconfluence of \textit{arbitrary} invariants, leaving the task of
proving or disproving \iconfluence to an automated model checker or
SMT solver. While I-confluence---like monotonicity and commutativity
(Section~\ref{sec:relatedwork})---is undecidable for arbitrary
programs, others have recently shown this alternative approach (e.g.,
in commutativity analysis~\cite{kohler-commutativity,redblue-new} and
in invariant generation for view serializable
transactions~\cite{writes-forest}) to be fruitful for restricted
languages. We view language design and more automated analysis as an
interesting area for more speculative research.

\minihead{Recency and session support} Our proposed invariants are
declarative, but a class of useful semantics---recency, or real-time
guarantees on reads and writes---are operational (i.e., they pertain
to transaction execution rather than the state(s) of the
database). For example, users often wish to read data that is
up-to-date as of a given point in time (e.g., ``read
latest''~\cite{pnuts} or linearizable~\cite{gilbert-cap}
semantics). While traditional isolation models do not directly address
these recency guarantees~\cite{adya-isolation}, they are often
important to programmers. Are these models \iconfluent? We can attempt
to simulate recency guarantees in \iconfluence analysis by logging the
result of all reads and any writes with a timestamp and requiring that
all logged timestamps respect their recency guarantees (thus treating
recency guarantees as invariants over recorded read/write execution
traces). However, this is a somewhat pointless exercise: it is well
known that recency guarantees are unachievable with transactional
availability~\cite{hat-vldb,gilbert-cap,davidson-survey}. Thus, if
application reads face these requirements, coordination is
required. Indeed, when application ''consistency'' means ``recency,''
systems cannot circumvent speed-of-light delays.

If users wish to ``read their writes'' or desire stronger ``session''
guarantees~\cite{bayou} (e.g., maintaining recency on a per-user or
per-session basis), they must maintain affinity or
``stickiness''~\cite{hat-vldb} with a given (set of) replicas. These
guarantees are also expressible in the \iconfluence model and do
not require coordination between different users' or sessions'
transactions.

\minihead{Physical and logical replication} We have used the concept
of replicas to reason about concurrent transaction execution. However,
as previously noted, our use of replicas is simply a formal device and
is independent of the actual concurrency control mechanisms at
work. Specifically, reasoning about replicas allows us to separate the
\textit{analysis} of transactions from their \textit{implementation}:
just because a transaction is executed with (or without) coordination
does not mean that all query plans or implementations require (or do
not require) coordination~\cite{hat-vldb}. However, in deciding on an
implementation, there is a range of design decisions yielding a
variety of performance trade-offs. Simply because an application is
\iconfluent does not mean that all implementations will perform
equally well. Rather, \iconfluence ensures that a coordination-free
implementation exists.

\minihead{Requirements and restrictions} Our techniques are predicated
on the ability to correctly and completely specify invariants and
inspect user transactions; without such a correctness specification,
for arbitrary transaction schedules, serializability is---in a
sense---the ``optimal'' strategy~\cite{kung1979optimality}. By casting
correctness in terms of admissible application states rather than as a
property of read-write schedules, we achieve a more precise statement
of coordination overheads. However, as we have noted, this does not
obviate the need for coordination in all cases. Finally, when full
application invariants are unavailable, individual, high-value
transactions may be amenable to optimization via \iconfluence
coordination analysis.

\section{Experiences With Coordination}
\label{sec:evaluation}

When achievable, coordination-free execution enables scalability
limited to that of available hardware. This is powerful: an
\iconfluent application can scale out without
sacrificing correctness, latency, or availability. In
Section~\ref{sec:bcc-practice}, we saw combinations of invariants
and transactions that were \iconfluent and others that were not. In this
section, we apply these combinations to the workloads of the
OLTP-Bench suite~\cite{oltpbench}, with a focus on the TPC-C
benchmark. Our focus is on the coordinaton required in order to
correctly execute each and the resulting, coordination-related
performance costs.

\subsection{TPC-C Invariants and Execution}
\label{sec:tpcc-invariants}

The TPC-C benchmark is the gold standard for database concurrency
control~\cite{oltpbench} both in research and in industry~\cite{tpcc},
and in recent years has been used as a yardstick for distributed
database concurrency control
performance~\cite{calvin,hstore,silo}. How much coordination does
TPC-C actually require a compliant execution?

The TPC-C workload is designed to be representative of a wholesale
supplier's transaction processing requirements. The workload has a number of
application-level correctness criteria that represent basic business
needs (e.g., order IDs must be unique) as formulated by the TPC-C
Council and which must be maintained in a compliant run. We can
interpret these well-defined ``consistency criteria'' as invariants and
subsequently use \iconfluence analysis to determine which transactions
require coordination and which do not.

Table~\ref{table:tpcc-invariants} summarizes the twelve invariants
found in TPC-C as well as their \iconfluence analysis results as
determined by Table~\ref{table:invariants}. We classify the invariants
into three broad categories: materialized view maintenance, foreign
key constraint maintenance, and unique ID assignment. As we discussed
in Section~\ref{sec:bcc-practice}, the first two categories are
\iconfluent (and therefore maintainable without coordination) because
they only regulate the \textit{visibility} of updates to multiple
records. Because these (10 of 12) invariants are \iconfluent under the
workload transactions, there exists some execution strategy that does
not use coordination. However, simply because these invariants are
\iconfluent does not mean that \textit{all} execution strategies will
scale well: for example, using locking would \textit{not} be
coordination-free.

\begin{table}[t!]
\definecolor{yesgray}{gray}{0.92}
\small
\begin{tabular}{|l|l|l|l|l|}
  \hline
  \textbf{\#} & \textbf{Informal Invariant Description} &
  \textbf{Type} & 
  \textbf{Txns} & \textbf{$\mathcal{I}$-C} \\\hline
  \rowcolor{yesgray}
  1 & {\scriptsize YTD wh sales = sum(YTD district sales)} & MV &
  P & Yes\\
  2 &  {\scriptsize Per-district order IDs are sequential} &$\mathrm{S_{ID}}$+FK  & N, D & No\\
  3 &  {\scriptsize New order IDs are sequentially assigned} & $\mathrm{S_{ID}}$  & N, D & No \\
  \rowcolor{yesgray}
  4 & {\scriptsize  Per-district, item order count = roll-up} & MV &
  N & Yes\\
  \rowcolor{yesgray}
  5 &  {\scriptsize Order carrier is set iff order is pending} & FK &
  N, D & Yes \\
  \rowcolor{yesgray}
  6 &  {\scriptsize Per-order item count = line item roll-up} & MV & N
  & Yes \\
  \rowcolor{yesgray}
  7 &  {\scriptsize Delivery date set iff carrier ID set} & FK & D & Yes \\
  \rowcolor{yesgray}
  8 &  {\scriptsize YTD wh = sum(historical wh)} & MV &
  D & Yes \\
  \rowcolor{yesgray}
  9 &  {\scriptsize YTD district = sum(historical district)} &
  MV & P & Yes \\
  \rowcolor{yesgray}
  10 &  {\scriptsize Customer balance matches expenditures} & MV & P, D & Yes \\
  \rowcolor{yesgray}
  11 &  {\scriptsize Orders reference New-Orders table}  & FK & N & Yes \\
  \rowcolor{yesgray}
  12 &  {\scriptsize Per-customer balance = cust. expenditures} & MV & P, D &
  Yes \\\hline

\end{tabular}
\caption{TPC-C Declared ``Consistency Conditions'' (3.3.2.x) and \iconfluence
  analysis results
  (Invariant type: MV: materialized view, $\mathbf{S_{ID}}$: sequential ID assignment, FK: foreign
  key; Transactions: N: New-Order, P: Payment, D: Delivery).}
\label{table:tpcc-invariants}
\end{table}

As one coordination-free execution strategy (which we implement in
Section~\ref{sec:evaltpcc}) that respects the foreign key and
materialized view invariants, we can use RAMP transactions, which
provide atomically visible transactional updates across servers
without relying on coordination for correctness~\cite{ramp-txns}. In
brief, RAMP transactions employ limited multi-versioning and metadata
to ensure that readers and writers can always proceed concurrently:
any client whose reads overlap with another client's writes to the
same item(s) can use metadata stored in the items to fetch any
``missing'' writes from the respective servers. A standard RAMP
transaction over data items suffices to enforce foreign key
constraints, while a RAMP transaction over commutative counters as
described in~\cite{ramp-txns} is sufficient to enforce the TPC-C
materialized view constraints.

Two of TPC-C's invariants are not \iconfluent with respect to the
workload transactions and therefore \textit{do} require
coordination. On a per-district basis, order IDs should be assigned
sequentially (both uniquely and sequentially, in the New-Order
transaction) and orders should be processed sequentially (in the
Delivery transaction). If the database is partitioned by warehouse (as
is standard~\cite{silo,calvin,hstore}), the former is a distributed
transaction (by default, $10\%$ of New-Order transactions span
multiple warehouses). The benchmark specification allows the latter to
be run asynchronously and in batch mode on a per-warehouse
(non-distributed) basis, so we, like others~\cite{calvin,silo}, focus
on New-Order. Including additional transactions like the read-only
Order-Status in the workload mix would increase performance due to the
transactions' lack of distributed coordination and (often
considerably) smaller read/write footprints.

\minihead{Avoiding New-Order Coordination} New-Order is not
\iconfluent with respect to the TPC-C invariants, so we can always
fall back to using serializable isolation. However, the per-district
ID assignment records (10 per warehouse) would become a point of contention,
limiting our throughput to effectively $\frac{100W}{RTT}$ for a
$W$-warehouse TPC-C benchmark with the expected $10\%$ distributed
transactions. Others~\cite{silo} (including us, in prior
work~\cite{hat-vldb}) have suggested disregarding consistency criteria
3.3.2.3 and 3.3.2.4, instead opting for unique but non-sequential ID
assignment: this allows inconsistency and violates the benchmark
compliance criteria.

During a compliant run, New-Order transactions must
coordinate. However, as discussed above, only the ID assignment
operation is non-I-confluent; the remainder of the operations in the
transaction can execute coordination-free. With some effort, we can
avoid distributed coordination. A na\"{\i}ve implementation might grab
a lock on the appropriate district's ``next ID'' record, perform
(possibly remote) remaining reads and writes, then release the lock at
commit time. Instead, as a more efficient solution, New-Order can
defer ID assignment until commit time by introducing a layer of
indirection. New-Order transactions can generate a temporary, unique,
but non-sequential ID (\texttt{tmpID}) and perform updates using this
ID using a RAMP transaction (which, in turn, handles the foreign key
constraints)~\cite{ramp-txns}. Immediately prior to transaction
commit, the New-Order transaction can assign a ``real'' ID by
atomically incrementing the current district's``next ID'' record
(yielding \texttt{realID}) and recording the \texttt{[tmpID, realID]}
mapping in a special ID lookup table. Any read requests for the
\texttt{ID} column of the Order, New-Order, or Order-Line tables can
be safely satisfied (transparently to the end user) by joining with
the ID lookup table on \texttt{tmpID}. In effect, the New-Order ID
assignment can use a nested atomic
transaction~\cite{atomictransactions} upon commit, and all
coordination between any two transactions is confined to a
single server.

\subsection{Evaluating TPC-C New-Order}
\label{sec:evaltpcc}

\begin{figure}
\hspace{1.5em}\includegraphics{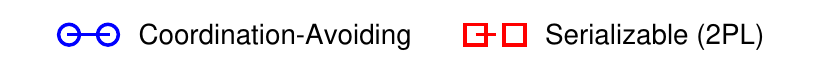}\vspace{-1em}
\includegraphics[width=\columnwidth]{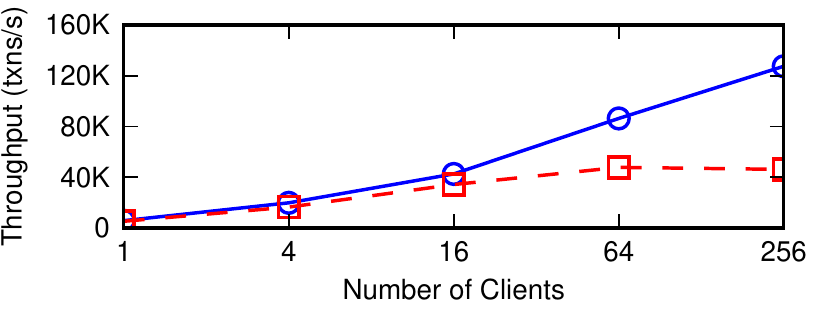}
\includegraphics[width=\columnwidth]{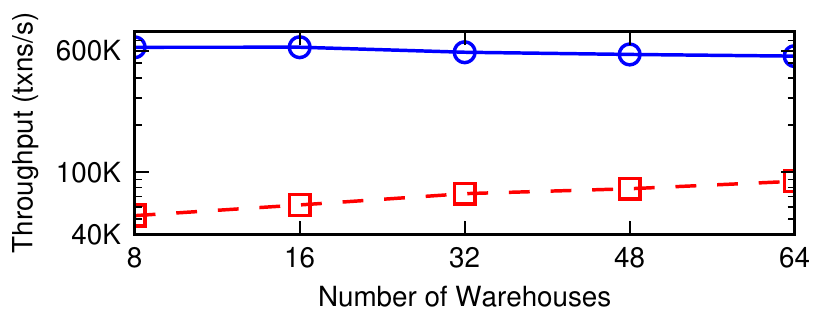}
\includegraphics[width=\columnwidth]{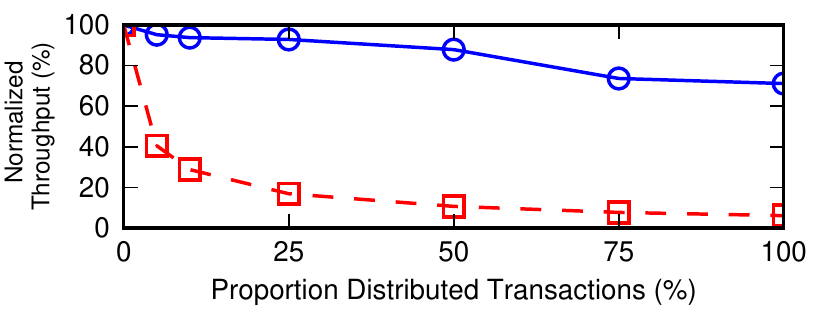}\vspace{-.5em}
\caption{TPC-C New-Order throughput across eight servers.}
\label{fig:clients}
\end{figure}

We subsequently implemented the above execution strategy in a distributed
database prototype to quantify the overheads associated with
coordination in TPC-C New-Order. In brief, the coordination-avoiding
query plan scales linearly to over 12.7M transactions per second on
$200$ servers while substantially outperforming distributed two-phase
locking. Our goal here is to demonstrate---beyond the microbenchmarks of
Section~\ref{sec:motivation}---that safe but judicious use of
coordination can have meaningful positive effect on performance.

\minihead{Implementation and Deployment} We employ a multi-versioned
storage manager, with RAMP-Fast transactions for snapshot reads and
atomically visible writes/``merge'' (providing a variant of regular
register semantics, with writes visible to later transactions after
commit)~\cite{ramp-txns} and implement the nested atomic transaction
for ID assignment as a sub-procedure inside RAMP-Fast's server-side
commit procedure (using spinlocks). We implement transactions as
stored procedures and fulfill the TPC-C ``Isolation Requirements'' by
using read and write buffering as proposed in~\cite{hat-vldb}. As is
common~\cite{calvin,abadi-vll,hstore,jones-dtxn}, we disregard
per-warehouse client limits and ``think time'' to increase load per
warehouse. In all, our base prototype architecture is similar to that
of~\cite{ramp-txns}: a JVM-based partitioned, main-memory, mastered
database.

For an apples-to-apples comparison with a coordination-intensive
technique within the same system, we also implemented textbook
two-phase locking (2PL)~\cite{bernstein-book}, which provides
serializability but also requires distributed coordination. We totally
order lock requests across servers to avoid deadlock, batching lock
requests to each server and piggybacking read and write requests on
lock request RPC. As a validation of our implementation, our 2PL
prototype achieves per-warehouse (and sometimes aggregate) throughput
similar to (and often in excess of) several recent serializable
database implementations (of both 2PL and other
approaches)~\cite{calvin,abadi-vll,hstore,jones-dtxn}.

By default, we deploy our prototype on eight EC2 \texttt{cr1.8xlarge}
instances in the Amazon EC2 \texttt{us-west-2} region (with
non-co-located clients) with one warehouse per server (recall there
are 10 ``hot'' district ID records per warehouse) and report the
average of three 120 second runs.

\minihead{Basic behavior} Figure~\ref{fig:clients} shows performance
across a variety of configurations, which we detail below. Overall,
the coordination-avoiding query plan far outperforms the serializable
execution. The coordination-avoiding query plan performs some
coordination, but, because coordination points are not distributed
(unlike 2PL), physical resources (and not coordination) are the bottleneck.

\miniheadit{Varying load} As we increase the number of clients, the
coordination-avoiding query plan throughput increases linearly, while
2PL throughput increases to $40$K transactions per second, then levels
off. As in our microbenchmarks in Section~\ref{sec:motivation}, the
former utilizes available hardware resources (bottlenecking
on CPU cycles at $640$K transactions per second), while the latter
bottlenecks on logical contention.

\miniheadit{Physical resource consumption} To understand the overheads
of each component in the coordination-avoiding query plan, we used JVM
profiling tools to sample thread execution while running at peak
throughput, attributing time spent in functions to relevant modules
within the database implementation (where possible):

\vspace{-.5em}
\begin{center}
\centering
\small
\setlength{\fboxsep}{4pt}
\fbox{
\begin{tabular}{l r}
\textbf{Code Path} & \textbf{Cycles}\\
Storage Manager (Insert, Update, Read) & 45.3\%\\
Stored Procedure Execution & 14.4\%\\
RPC and Networking & 13.2\%\\
Serialization & 12.6\%\\
ID Assignment Synchronization (spinlock contention) & 0.19\%\\
Other & 14.3\%\\
\end{tabular}}
\end{center}
\vspace{-.5em}

The coordination-avoiding prototype spends a large portion of
execution in the storage manager, performing B-tree modifications and
lookups and result set creation, and in RPC/serialization. In contrast
to 2PL, the prototype spends less than $0.2\%$ of time coordinating,
in the form of waiting for locks in the New-Order ID assignment; the
(single-site) assignment is fast (a linearizable integer increment and
store, followed by a write and fence instruction on the spinlock), so
this should not be surprising. We observed large throughput penalties
due to garbage collection (GC) overheads (up to 40\%)---an unfortunate
cost of our highly compact (several thousand lines of Scala),
JVM-based implementation. However, even in this current prototype,
physical resources are the bottleneck---not coordination.

\miniheadit{Varying contention} We subsequently varied the number of
``hot,'' or contended items by increasing the number of warehouses on
each server. Unsurprisingly, 2PL benefits from a decreased
contention, rising to over $87$K transactions per second with $64$
warehouses. In contrast, our coordination-avoiding implementation is
largely unaffected (and, at $64$ warehouses, is even negatively
impacted by increased GC pressure). The coordination-avoiding query
plan is effectively agnostic to read/write contention.

\miniheadit{Varying distribution} We also varied the percentage of
distributed transactions. The coordination-avoiding query plan
incurred a $29\%$ overhead moving from no distributed transactions to
all distributed transactions due to increased serialization overheads
and less efficient batching of RPCs. However, the 2PL implementation decreased in throughput by over $90\%$ (in line
with prior results~\cite{abadi-vll,calvin}, albeit exaggerated here
due to higher contention) as more requests stalled due to coordination
with remote servers.

\minihead{Scaling out} Finally, we examined our prototype's scalability,
again deploying one warehouse per server. As Figure~\ref{fig:scaleout}
demonstrates, our prototype scales linearly, to over 12.74 million
transactions per second on 200 servers (in light of our earlier
results, and, for economic reasons, we do not run 2PL at this
scale). Per-server throughput is largely constant after 100 servers,
at which point our deployment spanned all three \texttt{us-west-2}
datacenters and experienced slightly degraded per-server performance.
While we make use of application semantics, we are unaware of
any other compliant multi-server TPC-C implementation that has
achieved greater than 500K New-Order transactions per
second~\cite{calvin,jones-dtxn,hstore,abadi-vll}.

\minihead{Summary} We present these quantitative results as a proof of
concept that executing even challenging workloads like TPC-C that
contain complex integrity constraints are not necessarily at odds with
scalability if implemented in a coordination-avoiding
manner. Distributed coordination need not be a bottleneck for all
applications, even if conflict serializable execution indicates
otherwise. Coordination avoidance ensures that physical
resources---and not logical contention---are the system bottleneck
whenever possible.

\begin{figure}
\includegraphics[width=\columnwidth]{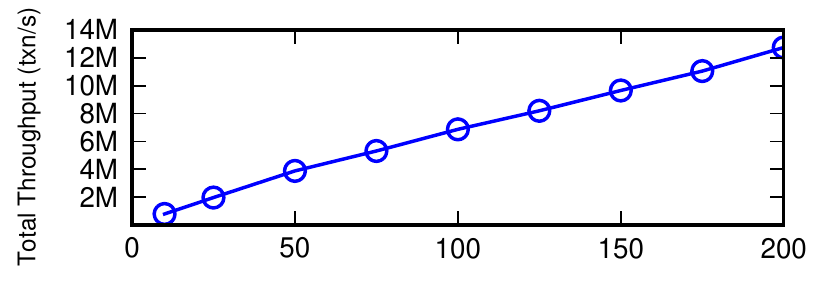}\vspace{-.5em}
\includegraphics[width=\columnwidth]{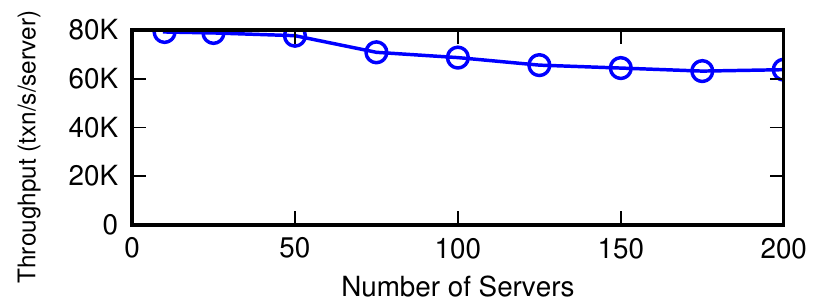}\vspace{-.5em}
\caption{Coordination-avoiding New-Order scalability.}
\label{fig:scaleout}
\end{figure}

\subsection{Analyzing Additional Applications}

These results begin to quantify the effects of coordination-avoiding
concurrency control. If considering \textit{application-level}
invariants, databases only have to pay the price of coordination when
necessary. We were surprised that the ``current industry standard for
evaluating the performance of OLTP systems''~\cite{oltpbench} was so
amenable to coordination-avoiding execution---at least for compliant
execution as defined by the official TPC-C specification.

For greater variety, we also studied the workloads of the recently
assembled OLTP-Bench suite~\cite{oltpbench}, performing a similar
analysis to that of Section~\ref{sec:tpcc-invariants}. We found (and
confirmed with an author of~\cite{oltpbench}) that for nine of
fourteen remaining (non-TPC-C) OLTP-Bench applications, the workload
transactions did not involve integrity constraints (e.g., did not
modify primary key columns), one (\texttt{CH-benCHmark}) matched
TPC-C, and two specifications implied (but did not explicitly state) a
requirement for unique ID assignment (\texttt{AuctionMark}'s
\texttt{new-purchase} order completion, \texttt{SEATS}'s
\texttt{NewReservation} seat booking; achievable like TPC-C order
IDs). The remaining two benchmarks, \texttt{sibench} and
\texttt{smallbank} were specifically designed (by an author of this
paper) as research benchmarks for serializable isolation. Finally, the
three ``consistency conditions'' required by the newer TPC-E benchmark
are a proper subset of the twelve conditions from TPC-C considered
here (and are all materialized counters). It is possible (even likely)
that these benchmarks are underspecified, but according to official
specifications, TPC-C contains the most coordination-intensive
invariants among all but two of the OLTP-Bench workloads.

Anecdotally, our conversations and experiences with real-world
application programmers and database developers have not identified
invariants that are radically different than those we have studied
here. A simple thought experiment identifying the invariants required
for a social networking site yields a number of invariants but none
that are particularly exotic (e.g., username uniqueness, foreign key
constraints between updates, privacy
settings~\cite{pnuts,ramp-txns}). Nonetheless, we view the further study of
real-world invariants to be a necessary area for future
investigation. In the interim, these preliminary results hint at what
is possible with coordination-avoidance as well as the costs of
coordination if applications are not \iconfluent.

\section{Related Work}
\label{sec:relatedwork}

Database system designers have long sought to manage the trade-off
between consistency and coordination. As we have discussed,
serializability and its many implementations (including lock-based,
optimistic, and pre-scheduling
mechanisms)~\cite{silo,bernstein-book,tamer-book,hstore,gray-virtues,calvin,eswaran-consistency,sdd1}
are sufficient for maintaining application correctness. However,
serializability is not always necessary: as discussed in
Section~\ref{sec:intro}, serializable databases do not allow certain
executions that are correct according to application semantics.  This
has led to a large class of application-level---or
semantic---concurrency control models and mechanisms that admit greater
concurrency. There are several surveys on this topic, such
as~\cite{tamer-book,ic-survey}, and, in our solution, we integrate
many concepts from this literature.

\minihead{Commutativity} One of the most popular alternatives to
serializability is to exploit \textit{commutativity}: if transaction
return values (e.g., of reads) and/or final database states are
equivalent despite reordering, they can be executed
simultaneously~\cite{weihl-thesis,kohler-commutativity,redblue}. Commutativity
is often sufficient for correctness but is not necessary. For example,
if an analyst at a wholesaler creates a report on daily cash flows,
any concurrent sale transactions will \textit{not} commute with the
report (the results will change depending on whether the sale completes
before or after the analyst runs her queries). However, the report
creation is \iconfluent with respect to, say, the invariant that every
sale in the report references a customer from the customers
table. \cite{kohler-commutativity,lamport-audit} provide additional
examples of safe non-commutativity.

\minihead{Monotonicity and Convergence} The CALM
Theorem~\cite{ameloot-calm} shows that monotone programs exhibit
deterministic outcomes despite re-ordering. CRDT objects~\cite{crdt}
similarly ensure convergent outcomes that reflect all updates made to
each object. These outcome determinism and convergence guarantees are
useful \textit{liveness} properties~\cite{schneider-concurrent} (e.g.,
a converged CRDT OR-Set reflects all concurrent additions and
removals) but do not prevent users from observing inconsistent
data~\cite{redblue-new}, or \textit{safety} (e.g., the CRDT OR-Set
does not---by itself---enforce invariants, such as ensuring that no
employee belongs to two departments), and are therefore not sufficient
to guarantee correctness for all applications. Further understanding
the relationship between \iconfluence and CALM is an interesting area
for further exploration (e.g., as \iconfluence adds safety to
confluence, is there a natural extension of monotone logic that
incorporates \iconfluent invariants---say, via an ``invariant-scoped''
form of monotonicity?).

\minihead{Use of Invariants} A large number of database
designs---including, in restricted forms, many commercial databases
today---use various forms of application-supplied invariants,
constraints, or other semantic descriptions of valid database states
as a specification for application correctness
(e.g.,~\cite{korth-serializability,kemme-si-ic,garciamolina-semantics,ic-survey,ic-survey-two,decomp-semantics,redblue,writes-forest,davidson-survey,local-verification,redblue-new}). We
draw inspiration and, in particular, our use of invariants from this
prior work. However, we are not aware of related work that discusses
when coordination is strictly \textit{required} to enforce a given set
of invariants. Moreover, our practical focus here is primarily
oriented towards invariants found in SQL and from modern applications.

In this work, we provide a necessary and sufficient condition for
safe, coordination-free execution. In contrast with many of the
conditions above (esp. commutativity and monotonicity), we explicitly
require more information from the application in the form of
invariants (Kung and Papadimitriou~\cite{kung1979optimality} suggest
this is information is \textit{required} for general-purpose non-serializable
yet safe execution.)  When invariants are unavailable, many of these
more conservative approaches may still be applicable. Our use of
analysis-as-design-tool is inspired by this literature---in
particular,~\cite{kohler-commutativity}.

\minihead{Coordination costs} In this work, we determine when
transactions can run entirely concurrently and without
coordination. In contrast, a large number of alternative models
(e.g.,~\cite{garciamolina-semantics,korth-serializability,isolation-semantics,local-verification,kemme-si-ic,aiken-confluence,laws-order})
assume serializable or linearizable (and therefore coordinated)
updates to shared state. These assumptions are standard (but not
universal~\cite{ec-txns}) in the concurrent programming
literature~\cite{schneider-concurrent,laws-order}. (Additionally,
unlike much of this literature, we only consider a single set of
invariants per database rather than per-operation invariants.) For
example, transaction chopping~\cite{chopping} and later
application-aware
extensions~\cite{decomp-semantics,agarwal-consistency} decompose
transactions into a set of smaller transactions, providing increased
concurrency, but in turn require that individual transactions execute
in a serializable (or strict serializable) manner.  This reliance on
coordinated updates is at odds with our goal of coordination-free
execution. However, these alternative techniques are useful in
reducing the duration and distribution of coordination once it is
established that coordination is required.

\minihead{Term rewriting} In term rewriting systems, \iconfluence
guarantees that arbitrary rule application will not violate a given
invariant~\cite{obs-confluence}, generalizing Church-Rosser
confluence~\cite{termrewriting}. We adapt this concept and effectively
treat transactions as rewrite rules, database states as constraint
states, and the database merge operator as a special \textit{join}
operator (in the term-rewriting sense) defined for all
states. Rewriting system concepts---including
confluence~\cite{aiken-confluence}---have previously been integrated
into active database systems~\cite{activedb-book} (e.g., in triggers,
rule processing), but we are not familiar with a concept analogous to
\iconfluence in the existing database literature.

\minihead{Coordination-free algorithms and semantics} Our work is
influenced by the distributed systems literature, where
coordination-free execution across replicas of a given data item has
been captured as ``availability''~\cite{gilbert-cap,queue}. A large
class of systems provides availability via ``optimistic replication''
(i.e., perform operations locally, then
replicate)~\cite{optimistic}. We---like others~\cite{ec-txns}---adopt
the use of the merge operator to reconcile divergent database
states~\cite{bayou} from this literature. Both traditional database
systems~\cite{adya-isolation} and more recent
proposals~\cite{redblue-new, redblue} allow the simultaneous use of
``weak'' and ``strong'' isolation; we seek to understand \textit{when}
strong mechanisms are needed rather than an optimal implementation of
either. Unlike ``tentative update'' models~\cite{sagas}, we do not
require programmers to specify compensatory actions (beyond merge,
which we expect to typically be generic and/or system-supplied) and do
not reverse transaction commit decisions. Compensatory actions could
be captured under \iconfluence as a specialized merge procedure.

The CAP Theorem~\cite{gilbert-cap,pacelc} recently popularized the
tension between strong semantics and coordination and pertains to a
specific model (linearizability). The relationship between
serializability and coordination requirements has also been well
documented in the database literature~\cite{davidson-survey}. We
recently classified a range of weaker isolation models by
availability, labeling semantics achievable without coordination as
``Highly Available Transactions''~\cite{hat-vldb}. Our research here addresses
when particular \textit{applications} require coordination.

In our evaluation, we make use of our recent RAMP
transaction algorithms~\cite{ramp-txns}, which guarantee
coordination-free, atomically visible updates. RAMP transactions are
an \textit{implementation} of \iconfluent semantics (i.e., Read Atomic
isolation, used in our implementation for foreign key constraint
maintenance). Our focus in this paper is \textit{when} RAMP
transactions (and any other coordination-free or \iconfluent semantics)
are appropriate for applications.

\minihead{Summary} The \iconfluence property is a necessary and
sufficient condition for safe, coordination-free execution. Sufficient
conditions such as commutativity and monotonicity are useful in
reducing coordination overheads but are not always necessary. Here, we
explore the fundamental limits of coordination-free execution. To do
so, we explicitly consider a model without synchronous
communication. This is key to scalability: if, by default, operations
must contact a centralized validation service, perform atomic updates
to shared state, or otherwise communicate, then scalability will be
compromised. Finally, we only consider a single set of invariants for
the entire application, reducing programmer overhead without affecting
our \iconfluence results.

\ifextended

\section{Discussion and Future Work}
\label{sec:discussion}
\label{sec:futurework}

In this paper, we have focused on the problem of recognizing when it
is possible to avoid coordination. Here, we discuss extensions to our
approaches and outline areas for future work.

\minihead{\iconfluence as design tool} As we have discussed, if a
transaction is \iconfluent with respect to an invariant, there exists
a coordination-free algorithm for safely executing it.  For example, we
used an early version of \iconfluence analysis in the development of
our RAMP transactions: coordination-free, atomically visible
transactions across multiple partitions that are useful in several
other use cases like foreign key constraint
maintenance~\cite{ramp-txns}. As we showed in
Section~\ref{sec:bcc-practice}, insertion and cascading deletes are
\iconfluent under foreign key constraints, so, when seeking a highly
concurrent algorithm for this use case, we knew the search was not in
vain: \iconfluence analysis indicated there existed at least one safe,
coordination-free mechanism for the task. We see (and have continued
to use) the \iconfluence property as a useful tool in designing new
algorithms, particularly in existing, well-specified applications and
use cases (e.g., B-tree internals, secondary indexes).

\minihead{Amortizing coordination} We have analyzed conflicts on a
per-transaction basis, but it is possible to amortize the overhead of
coordination across multiple transactions. For example, the Escrow
transaction method~\cite{escrow} reduces coordination by allocating a
``share'' of non-\iconfluent operations between multiple
processes. For example, in a bank application, a balance of $\$100$
might be divided between five servers, such that each server can
dispense $\$20$ without requiring coordination to enforce a
non-negative balance invariant (servers can coordinate to ``refresh''
supply). In the context of our \cfreedom analysis, this is
similar to limiting the branching factor of the execution trace to a
finite factor. Adapting Escrow and alternative time-, versioned-, and
numerical- drift-based models~\cite{epsilon-divergence} is a promising
area for future work.

\minihead{System design} The design of full coordination-avoiding
database systems raises several interesting questions. For example,
given a set of \iconfluence results as in
Table~\ref{table:invariants}, does a coordination-avoiding system have
to know all queries in advance, or can it dynamically employ
concurrency primitives as queries are submitted? (Early experiences
suggest the latter.)  Revisiting heuristics- and statistics-based
query planning, specifically targeting physical layout,
choice of concurrency control, and recovery mechanisms appears
worthwhile. How should a system handle invariants that may change over
time? Is SQL the right target for language analysis? We view these
pragmatic questions as exciting areas for future work.

\fi

\section{Conclusion}
\label{sec:conclusion}

ACID transactions and associated strong isolation levels dominated the
field of database concurrency control for decades, due in large part
to their ease of use and ability to automatically guarantee
application correctness criteria. However, this powerful abstraction
comes with a hefty cost: concurrent transactions must coordinate in
order to prevent read/write conflicts that could compromise
equivalence to a serial execution. At large scale and, increasingly,
in geo-replicated system deployments, the coordination costs
necessarily associated with these implementations produce significant
overheads in the form of penalties to throughput, latency, and
availability. In light of these trends, we developed a formal
framework, called \fullnameconfluence, in which application invariants
are used as a basis for determining if and when coordination is
strictly necessary to maintain correctness. With this framework, we
demonstrated that, in fact, many---but not all---common database
invariants and integrity constraints are actually achievable without
coordination. By applying these results to a range of actual
transactional workloads, we demonstrated an opportunity to avoid
coordination in many cases that traditional serializable mechanisms
would otherwise coordinate. The order-of-magnitude performance
improvements we demonstrated via coordination-avoiding concurrency
control strategies provide compelling evidence that invariant-based
coordination avoidance is a promising approach to meaningfully scaling
future data management systems.

{ \iftrue \minihead{Acknowledgments} The authors would like to thank
  Peter Alvaro, Neil Conway, Shel Finkelstein, and Josh Rosen for
  helpful feedback on earlier versions of this work, Dan Crankshaw,
  Joey Gonzalez, Nick Lanham, and Gene Pang for various engineering
  contributions, and Yunjing Yu for sharing the Bobtail dataset. This
  research is supported in part by NSF CISE Expeditions Award
  CCF-1139158, LBNL Award 7076018, DARPA XData Award
  FA8750-12-2-0331, the NSF Graduate Research Fellowship (grant
  DGE-1106400), and gifts from Amazon Web Services, Google, SAP, the
  Thomas and Stacey Siebel Foundation, Adobe, Apple, Inc., Bosch,
  C3Energy, Cisco, Cloudera, EMC, Ericsson, Facebook, GameOnTalis,
  Guavus, HP, Huawei, Intel, Microsoft, NetApp, Pivotal, Splunk,
  Virdata, VMware, and Yahoo!. \fi \fi}

\scriptsize
\linespread{.98}
\selectfont

\bibliography{bcc} \bibliographystyle{abbrv}

\balance

\linespread{1}
\selectfont

\ifextended

\normalsize

\section*{APPENDIX A: Experimental Details}

\minihead{Microbenchmark experiment description} We implemented traditional two-phase locking and an optimized variant of two-phase locking on the experimental prototype described in Section~\ref{sec:evaluation}.

In two-phase locking, each client acquires locks one at a time, requiring a full round trip time (RTT) for every lock request. For an $N$ item transaction, locks are held for $2N+1$ message delays (the $+1$ is due to broadcasting the unlock/commit command to the participating servers).

Our optimized two-phase locking only uses one message delay (half RTT) to perform each lock request: the client specifies the entire set of items it wishes to modify at the start of the transaction (in our implementation, the number of items in the transaction and the starting item ID), and, once a server has updated its respective item, the server forwards the remainder of the transaction to the server responsible for the next write in the transaction. For an $N$-item transaction, locks are only held for $N$ message delays (the final server both broadcasts the unlock request to all other servers and also notifies the client), while a $1$-item transaction does not require distributed locking.

To avoid deadlock (which was otherwise common in this high-contention microbenchmark), our implementation totally orders any lock requests according to item and executes them sequentially (e.g., lock server $1$ then lock server $2$). Our implementation also piggybacks write commands along with lock requests, further avoiding message delays. Unlike the locking implementation used in Section~\ref{sec:evaluation}, since we are only locking one item per server, our microbenchmark code does not use a dynamic lock manager and instead associates a single lock with each item; this should further lower locking overheads.

On each server, our lock implementation uses Java \texttt{ReentrantLock}, which unfortunately means that, for all but $1$-item optimized 2PL, our implementation was unable to used fixed-size thread pools (in contrast with our Scala \texttt{Future}-based coordination-free runtime). Nevertheless, we do not believe that our locking implementation is the actual bottleneck in the distributed setting: coordination is.

We partitioned eight in-memory items (integers) across eight \texttt{cr1.8xlarge} Amazon EC2 instances with clients located on a separate set of \texttt{cr1.8xlarge} instances. Figure~\ref{fig:micro} reported in Section~\ref{sec:motivation} depicts results for the coordination-free implementation and the optimized two-phase locking case. Figure~\ref{fig:micro-all} in this second depicts all three algorithms. Unsurprisingly, two-phase locking performs worse than optimized two-phase locking, but both incur substantial penalties due to coordination delay over the network.

\begin{figure}
\includegraphics[width=\columnwidth]{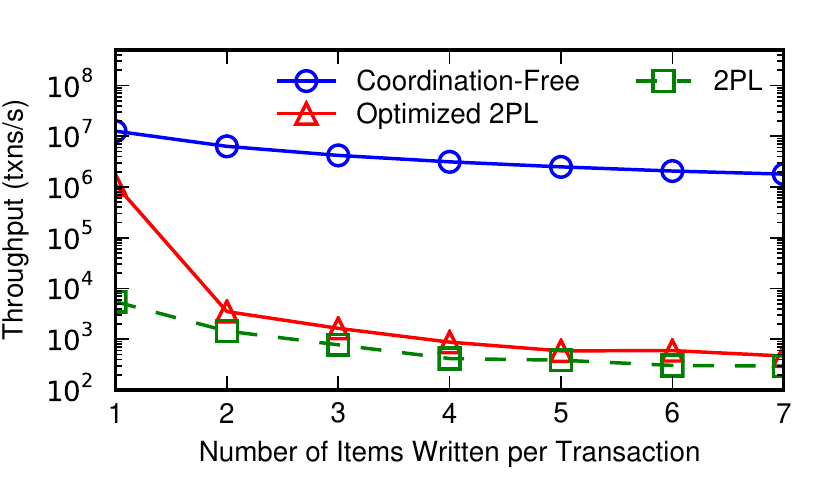}
\caption{Additional performance details for microbenchmark performance of conflicting and non-conflicting
  transactions.}
\label{fig:micro-all}
\end{figure}

\minihead{Trace-based simulation description} We simulate traditional two-phase commit and decentralized two-phase commit, using network models derived from existing studies. Our simulation is rather straightforward, but we make several optimizations to improve the throughput of each algorithm. First, we assume that transactions are pipelined, so that each server can \texttt{prepare} immediately after it has \texttt{commit}ted the prior transaction. Second, our pipelines are ideal in that we do not consider deadlock: only one transaction \texttt{prepare}s at a given time. Third, we do not consider the cost of local processing of each transaction: throughput is determined entirely by communication delay.

While this study is based solely on reported latencies, deployment
reports corroborate our findings. For example, Google's F1 uses
optimistic concurrency control via WAN with commit latencies of $50$
to $150$~ms. As the authors discuss, this limits throughput to between
$6$ to $20$ transactions per second per data
item~\cite{f1}. Megastore's average write latencies of $100$ to
$400$~ms suggest similar throughputs to those that we have
predicted~\cite{megastore}. Again, \textit{aggregate} throughput may
be greater as multiple 2PC rounds for disjoint sets of data items may
safely proceed in parallel. However, \textit{worst-case} access
patterns will indeed greatly limit scalability.

\section*{APPENDIX B: \iconfluence proof}

\begin{lemma}\label{lemma:force} 
Given a set of transactions $T$ and invariants $I$, a globally $I$-valid, coordination-free, transactionally available, and convergent system is able to produce any $I$-$T$-reachable state $S_i$.
 \end{lemma} \begin{proof}{Lemma~\ref{lemma:force}}
Let $\alpha_i$ represent a partially ordered sequence of transactions $T_i$ and merge procedure invocations $M_i$ (call this a \textit{history}) starting from $S_0$ that produces $S_i$,

We $\textsc{replay}$ the history $\alpha_i$ on a set of servers as follows. Starting from the initial state $S_0$, we traverse the partial order according to what amounts to a topological sort. Initially, we mark all operations (transactions or merges) in $\alpha_i$ as \textit{not done}. We begin by executing all transactions $T_i$ in $\alpha_i$ that have no predeceding operations in $\alpha_i$. For each transaction $t \in T_i$, we execute $t$ on a server that is unable to communicate with any other server. Upon transaction commit, we merge each replica's modifications into a separate server.  (Recall that, because $S_i$ is $I$-$T$-reachable, each transaction in $\alpha_i$ is an $I$-valid transformation and must either eventually commit or abort itself to preserve transactional availability, and, due to coordination-freedom, the result of the execution is dependent solely on its input---in this case, $S_0$.) We subsequently mark each $t \in T_i$ as \textit{done} in $\alpha_0$ and denote the resulting server as $s_t$.\footnote{Recall from Section~\ref{sec:model} that we consider arbitrary sets of servers. We could likely (at least in certain cases) be more parsimonius with our use of servers in this proof at some cost to complexity.} 

Next, we repeatedly select an operation $o_i$ from $\alpha_i$ that is marked as \textit{not done} but whose preceding operations are all marked as \textit{done}.

If $o_i$ is a transaction with preceding operation $o_j$ on corresponding server $s_j$, we partition $s_j$ and $s_i$ and another server containing state $S_0$ such that $s_j$ and $s_i$ can communicate with each other but cannot communicate with any other server. Under convergent execution, $s_j$ and $s_i$ must eventually contain the same state, merged to $s_j$ (given that $s_j \sqcup S_0$ is defined in our model to be $s_j$). Following convergence, we partition $s_j$ and $s_i$ so they cannot communicate. We subsequently execute $o_i$ on $s_j$. Again, $o_i$ must either commit or abort itself to preserve transactional availability, and its behavior is solely dependent on its input due to coordination-freedom.

If $o_i$ is a merge procedure with preceding operations $o_j$ and $o_k$ on corresponding servers $s_j$ and $s_k$, we produce servers $s_{j'}$ and $s_{k'}$ as above, by partitioning $s_{j'}$ and $s_{j}$ and, respectively, $s_{k}$ and $s_{k'}$, waiting until convergence, then repartitioning each. Subsequently, we place $s_{j'}$ and $s_{k'}$ in the same partition, forcing the merge ($o_i$) of these states via the convergence requirement.

When all operations in $\alpha_i$ are marked as \textit{done}, the final operation we have performed will produce server containing state $S_i$. We have effectively (serially) traversed the history, inducing the series of transactions by triggering partitions while requiring commits due to transactional availability and merges due to our pair-wise convergence.
\end{proof}

We proceed to prove Theorem~\ref{theorem:necessary} from Section~\ref{sec:ic-result}.

\begin{proof}{Theorem~\ref{theorem:necessary}}
($\Leftarrow$) We begin with the simpler proof, which is by
  construction. Assume a set of transactions $T$ are \iconfluent with
  respect to an invariant $I$. Consider a system in which each server
  executes the transactions it receives against a replica of its current state and checks whether or not the resulting state is $I$-valid. If
  the resulting state is $I$-valid, the replica commits the
  transaction and its mutations to the state. If not, the replica
  aborts the transaction. Servers opportunistically exchange copies of
  their local states and merge them. No individual replica will
  install an invalid state upon executing transactions, and, because
  $T$ is \iconfluent under $I$, the merge of any two $I$-valid replica
  states from individual servers (i.e., $I$-$T$-reachable states) as
  constructed above is $I$-valid. Therefore, the converged database
  state will be $I$-valid. Transactional availability, convergence,
  and global $I$-validity are all maintained via coordination-free
  execution.

($\Rightarrow$) Assume a system $M$ guarantees globally $I$-valid
  operation for set of transactions $T$ and invariant $I$ with
  \cfreedom, transactional availability, and convergence, but $T$ is
  not $I$-confluent. Then there exist two $I$-$T$-reachable states $S_1$ and $S_2$ with common ancestor $I$-$T$-reachable state $S_a$ such that, by definition, $I(S_1)$ and $I(S_2)$ are true, but $I(S_1 \sqcup S_2)$ is false. \footnote{We may be able to apply Newman's
  lemma and only consider single-transaction divergence (in the case 
  of convergent and therefore ``terminating'' 
  executions)~\cite{obs-confluence,termrewriting}, but this is not 
  necessary for our results.\label{fn:newman-note}}

Consider two executions of system $M$, $\epsilon_1$ and $\epsilon_2$. In each execution, we begin by forcing $M$ to produce a server containing $S_c$ (via \textsc{replay} in Lemma~\ref{lemma:force}). In $\epsilon_1$, we subsequently \textsc{replay} the history $\alpha_1$ starting from $S_c$. In $\epsilon_2$, we subsequently \textsc{replay} the history $\alpha_2$ starting from $S_c$. Call $T_{f1}$ and $T_{f2}$ the final (set of) transactions that produced each of $S_1$ and $S_2$ (that is, the set of transactions in each execution that are not followed by any other transaction). Under $\alpha_1$ and $\alpha_2$, all transactions in each of $T_{f1}$ and $T_{f2}$ will have committed to maintain transactional availability, their end result will be equivalent to the result in $\alpha_1$ and $\alpha_2$ due to the coordination-freedom property, and $S_1$ and $S_2$ are both $I$-valid, by assumption.

We now consider a third execution, $\alpha_3$. $\alpha_3$ proceeds to independently \textsc{replay} $\alpha_1$ and $\alpha_2$ but does not execute or proceed further in the partial order than any element of $T_{f1}$ or $T_{f2}$; we consider these specially:

Once $\alpha_3$ has forced $M$ to \textsc{replay} other operations, we force it to \textsc{replay} these final operations beginning from $T_{f1}$ and $T_{f2}$. If we \textsc{replay} these transactions and $M$ commits them, we can \textsc{replay} the remainder of the operations in $\alpha_1$ and $\alpha_2$. In this case, due to the coordination-freedom property, $M$ will produce two servers $s_i$ and $s_j$ containing states $S_1$ and $S_2$. When we partition $s_i$ and $s_j$ such that they can communicate with each other but cannot communicate with any other servers, $s_i$ and $s_j$ must eventually converge, violating global $I$-validity. On the other hand, if $M$ aborts one or more of the transactions in $T_{f1}$ and $T_{f2}$, $M$ will not produce $s_i \sqcup s_j$, but, from the perspective of each server $s_{p1}$ executing a transaction in $T_{f1}$, this execution is indistinguishable from $\alpha_1$, and, from the perspective of each server $s_{p2}$ executing a transaction in $T_{f2}$, is indistinguishable from $\alpha_2$, a contradiction.

Therefore, to preserve transactional availability, $M$ must sacrifice one of global validity (by allowing the invalid merge), convergence (by never   merging), or \cfreedom (by requiring changes to transaction behavior).
\end{proof}

\section*{APPENDIX C: \iconfluence Analysis}

In this section, we more formally demonstrate the \iconfluence of invariants and operations discussed in Section~\ref{sec:merge}. Our goals in this section are two-fold. First, we have found the experience of formally proving \iconfluence to be instructive in understanding these combinations (beyond less formal arguments made in the body text for brevity and intuition). Second, we have found \iconfluence proofs to take on two general structures that, via repetition and in and variations below, may prove useful to the reader. In particular, the structure of our \iconfluence proofs takes one of two forms:

\begin{itemize}

\item To show a set of transactions are not \iconfluent with respect to an invariant $I$, we use proof by counterexample: we present two $I$-$T$-reachable states with a common ancestor that, when merged, are not $I$-valid. 

\item To show a set of transactions are \iconfluent with respect to an invariant $I$, we use proof by contradiction: we show that, if a state $S$ is not $I$-valid, merging two $I$-$T$-reachable states with a common ancestor state to produce $S$ implies either one or both of $S_1$ or $S_2$ must not be $I$-valid.

\end{itemize}

These results are not exhaustive, and there are literally infinite combinations of invariants and operations to consider. Rather, the seventeen examples below serve as a demonstration of what can be accomplished via \iconfluence analysis.

Notably, the negative results below use fairly simple histories consisting of a single transaction divergence. As we hint in Footnote~\ref{fn:newman-note}, it is possible that a large class of relevant invariant-operation pairs only depend on single-transaction divergence. Nevertheless, we decided to preserve the more general formulation of \iconfluence (accounting for arbitrary $I$-$T$-reachable states) to account for more pathological (perhaps less realistic, or, if these results are any indication, less commonly encountered) behaviors that only arise during more complex divergence patterns.

We introduce additional formalism as necessary. To start, unless otherwise specified, we use the set union merge operator. We denote version $i$ of item $x$ as $x_i$ and a write of version $x_i$ with value $v$ as $w(x_i=v)$.

\begin{claim}[Writes are \iconfluent with respect to per-item equality constraints]
\label{claim:eq-proof}
Assume writes are not \iconfluent with respect to some per-item equality constraint $i=c$, where $i$ is an item and $c$ is a constant. By definition, there must exist two $I$-$T$-reachable states $S_1$ and $S_2$ with common ancestor state such that $I(S_1) \rightarrow true$ and $I(S_2) \rightarrow true$ but $I(S_1) \sqcup S_2) \rightarrow false$; therefore, there exists a version $i_n$ in $S_1 \sqcup S_2$ such that $i_n \neq c$, and, under set union, $i_n \in S_1$, $i_n \in S_2$, or both. However, this would imply that $I(S_1) \rightarrow false$ or $I(S_2) \rightarrow false$ (or both), a contradiction.
\end{claim} 

\begin{claim}[Writes are \iconfluent with respect to per-item inequality constraints] \label{claim:neq} The proof follows almost identically to the proof of Claim, but for an invariant of the form $i\neq c$~\ref{claim:eq-proof}.  \end{claim}

\begin{claim}[Writing arbitrary values is not \iconfluent with respect to multi-item uniqueness constraints]
\label{claim:sunique}
Consider the following transactions:
\begin{align*}
T_{1u}&\coloneqq w(x_a=v);~commit\\
T_{2u}&\coloneqq w(x_b=v);~commit
\end{align*}
and uniqueness constraint on records:
$$I_u(D) = \{\textrm{values in D are unique}\}$$
Now, an empty database trivially does not violate uniqueness constraints ($I_u(D_s=\{\})\rightarrow true$), and adding individual versions to the separate empty databases is also valid:
\begin{align*}
T_{1u}(\{\})&=\{x_a=v\},~I_u(\{x_a=v\}) \rightarrow true\\
T_{2u}(\{\})&=\{x_b=v\},~I_u(\{x_b=v\}) \rightarrow true
\end{align*}
However, merging these states results in invalid state:
$$I_u(\{x_a=v\}\sqcup \{x_b=v\} = \{x_a=v, x_b=v\}) \rightarrow false$$
Therefore, $\{T_{1u}, T_{2u}\}$ is not \iconfluent under $I_s$.
\end{claim} 

For the next proof, we consider a model as suggested in Section~\ref{sec:bcc-practice} where replicas are able to generate unique (but not arbitrary (!)) IDs (in the main text, we suggested the use of a replica ID and sequence number). In the following proof, to account for this non-deterministic choice of unique ID, we introduce a special $nonce()$ function and require that, $nonce()$ return unique values for each replica; that is, $\sqcup$ is not defined for replicas on which independent invocations of $nonce()$ return the same value.

\begin{claim}[Assigning values by $nonce()$ is \iconfluent with respect to multi-item uniqueness constraints]
\label{claim:nunique}

Assume that assigning values by $nonce()$ is not \iconfluent with respect to some multi-item uniqueness invariant:
$$I(D) = \forall c \in dom(D), \{|\{x \in D \mid x = c\}| \leq 1 \}$$
By definition, there must exist two $I$-$T$-reachable states with a common ancestor reached by executing nonce-generating transactions (of the form $T_i=[w(x_i = nonce())]$), $S_1$ and $S_2$ such that $I(S_1) \rightarrow true$ and $I(S_2) \rightarrow true$ but $I(S_1 \sqcup S_2) \rightarrow false$.

Therefore, there exist two versions $i_a, i_b$ in $S_1 \sqcup S_2$ such that $i_a$ and $i_b$ (both generated by $nonce()$) are equal in value. Under set union, this means $i_a \in S_1$ and $i_b \in S_2$ ($i_a$ and $i_b$ both cannot appear in $S_1$ or $S_2$ since it would violate those states' $I$-validity). Because replica states grow monotonically under set merge and $S_1$ and $S_2$ differ, they must be different replicas. But $nonce()$ cannot generate the same values on different replicas, a contradiction. \end{claim}

\begin{claim}[Writing arbitrary values are not \iconfluent with respect to sequentiality constraints]
\label{claim:nseq-insert}
Consider the following transactions:
\begin{align*}
T_{1s}&\coloneqq w(x_a=1);~commit\\
T_{2s}&\coloneqq w(x_b=3);~commit
\end{align*}
and the sequentiality constraint on records:
$$I_s(D) =\{max(r\in D)-min(r\in D) = |D|+1\} \vee \{|D|=0\}$$
Now, $I_s$ holds over the empty database ($I_s(\{\}) \rightarrow true$), while inserting sequential new records into independent, empty replicas is also valid:
\begin{align*}
T_{1s}(\{\})&=\{x_a=1\},~I_u(\{x_a=1\} \rightarrow true\\
T_{2s}(\{\})&=\{x_b=3\},~I_u(\{x_b=3\} \rightarrow true
\end{align*}
However, merging these states results in invalid state:
$$I_s(\{x_a=1\}\sqcup \{x_b=3\} = \{x_a=1, x_b=3\}) \rightarrow false$$
Therefore, $\{T_{1s}, T_{2s}\}$ is not \iconfluent under $I_s$.
\end{claim}

To discuss foreign key constraints, we need some way to \textit{refer} to other records within the database. There are a number of ways of formalizing this; we are not picky and, here, refer to a field $f$ within a given version $x_i$ as $x_i.f$.

\begin{claim}[Insertion is \iconfluent with respect to foreign key constraints]
\label{claim:fk-insert}
  Assume that inserting new records is not \iconfluent with respect to some foreign key constraint $I(D) = \{\forall r_f \in D$ such that $r_f.g \neq null$, $\exists r_t \in D$ such that $r_f.g = r_t.h\}$ (there exists a foreign key reference between fields $g$ and $h$). By definition, there must exist two $I$-$T$-reachable states $S_1$ and $S_2$ with a common ancestor reachable by executing transactions performing insertions such that $I(S_1) \rightarrow true$ and $I(S_2) \rightarrow true$ but $I(S_1 \sqcup S_2) \rightarrow false$; therefore, there exists some version $r_1 \in S_1 \sqcup S_2$ such that $r_1.f \neq null$ but $\nexists r_2 \in S_1 \sqcup S_2$ such that $r_1.g = r_2.h$. Under set union, $r_1$ must appear in either $S_1$ or $S_2$ (or both), and, for each set of versions in which it appears, because $S_1$ and $S_2$ are both $I$-valid, they must contain an $r_3$ such that $r_1.f = r_3.h$. But, under set union, $r_3.h$ should also appear in $S_1 \sqcup S_2$, a contradiction.
\end{claim}

For simplicity, in the following proof, we assume that deleted elements remain deleted under merge. In practice, this can be accomplished by tombstoning records and, if required, using counters to record the number of deletions and additions~\cite{crdt}. We represent a deleted version $x_d$ by $\neg x_b$.

\begin{claim}[Concurrent deletion and insertion is not \iconfluent with respect to foreign key constraints]
\label{claim:fk-delete}
 Consider the following transactions:
\begin{align*}
T_{1f}&\coloneqq w(x_a.g=1);~commit\\
T_{2f}&\coloneqq delete(x_b);~commit
\end{align*}
and the foreign key constraint:
$$I_f(D) = \{\forall r_f \in D, r_f.g \neq null,~\exists r_t \in D \suchthat \neg r_t \notin D \mathand r_f.g = r_t.h\}$$
Foreign key constraints hold over the initial database $S_i=\{x_b.h=1\}$ ($I_u(S_i) \rightarrow true$) and on independent execution of $T_a$ and $T_b$:
\begin{align*}
T_{1f}(\{x_b.h=1\})&=\{x_a.g=1, x_b.h=1\},~I_f(\{x_a=1\}) \rightarrow true\\
T_{2f}(\{x_b.h=1\})&=\{x_b.h=1, \neg {x_b}\}~I_f(\{x_b.h=1, \neg {x_b}\}) \rightarrow true
\end{align*}
However, merging these states results in invalid state: 
$$I_f(\{x_a.g=1\}\sqcup \{x_b.h=1, \neg {x_b}\}) \rightarrow false$$
Therefore, $\{T_{1f}, T_{2f}\}$ is not \iconfluent under $I_f$.\end{claim}

We denote a casading delete of all records that reference field $f$ with value $v$ ($v$ a constant) as $cascade(f=v)$.

\begin{claim}[Cascading deletion and insertion are \iconfluent with respect to foreign key constraints]
\label{claim:fk-cascade}
Assume that cacading deletion and insertion of new records are not \iconfluent with respect to some foreign key constraint $I(D) = \{\forall r_f \in D$ such that $r_f.g \neq null$, $\exists r_t \in D$ such that $r_f.g = r_t.h$ if $cascade(h=r_f.g \neq v)\}$ (there exists a foreign key reference between fields $g$ and $h$ and the corresponding value for field $h$ has not been deleted-by-cascade). By definition, there must exist two $I$-$T$-reachable states $S_1$ and $S_2$ with common ancestor reachable by performing insertions such that $I(S_1) \rightarrow true$ and $I(S_2) \rightarrow true$ but $I(S_1 \sqcup S_2) \rightarrow false$; therefore, there exists some version $r_1 \in S_1 \sqcup S_2$ such that $r_1.f \neq null$ but $\nexists r_2 \in S_1 \sqcup S_2$ such that $r_1.g = r_2.h$. From the proof of Claim~\ref{claim:fk-insert}, we know that insertion is \iconfluent, so the absence of $r_2$ must be due to some cascading deletion. Under set union, $r_1$ must appear in exactly one of $S_1$ or $S_2$ (if $r_1$ appeared in both, there would be no deletion, a contradiction since we know insertion is \iconfluent). For the state $S_j$ in which $r_1$ does not appear (either $S_1$ or $S_2$), $S_j$ must include $cascade(h=r_1.g)$. But, if $cascade(h=r_1.g) \in S_j$, $cascade(h=r_1.g)$ must also be in $S_i \sqcup S_j$, a contradiction and so $S_i \sqcup S_j \rightarrow true$, a contradiction.
\end{claim}

We define a ``consistent'' secondary index invariant as requiring that, when a record is visible, its secondary index entry should also be visible. This is similar to the guarantees provided by Read Atomic isolation~\cite{ramp-txns}. For simplicity, in the following proof, we only consider updates to a single indexed attribute $attr$, but the proof is easily generalizable to multiple index entries, insertions, and deletion via tombstones. We use last-writer wins for index entries.

\begin{claim}[Updates are \iconfluent with respect to consistent secondary indexing] \label{claim:indexing}
  Assume that updates to records are not \iconfluent with respect a secondary index constraint on attribute $attr$:\vspace{.5em}

\noindent $I(D) = \{\forall r_f \in D$ such that $r_f.attr \neq null$ and $f$ is the highest version of $r\in D$, $\exists r_{idx} \in D$ such that $r_f \in r_{idx}.entries$ (all entries with non-null $attr$ are reflected in the secondary index entry for $attr$)\vspace{.5em}

  Represent an update to record $r_x$ as $\{w(r_x)$ and, if $r_x.attr \neq null$, also $r_{idx}.entries.add(r_x)$, else $r_{idx}.entries.delete(r_x)\}$.\vspace{.5em}

  By definition, there must exist two $I$-$T$-reachable states $S_1$ and $S_2$ with common ancestors reachable by performing insertions $S_1$ and $S_2$ such that $I(S_1) \rightarrow true$ and $I(S_2) \rightarrow true$ but $I(S_1 \sqcup S_2) \rightarrow false$; therefore, there exists some version $r_1 \in S_1 \sqcup S_2$ such that $r_1.attr \neq null$ but $\nexists r_{idx} \in S_1 \sqcup S_2$ or $\exists r_{idx} \in S_1 \sqcup S_2$ but $r_1 \notin r_{idx}.entries$. In the former case, $r_{idx} \notin S_1$ or $S_2$, but $r_1 \in S_1$ or $r_1 \in S_2$, a contradiction. The latter case also produces a contradiction: if $r_1 \in S_1$ or $r_1 \in S_2$, it must appear in $r_{idx}$, a contradiction.  \end{claim}

In our formalism, we can treat materialized views as functions over database state $f(D) \rightarrow c$.

\begin{claim}[Updates are \iconfluent with respect to materialized view maintenance] The proof is relatively straightforward if we treat the materialized view record(s) $r$ as having a foreign key relationship to any records in the domain of the function (as in the proof of Claim~\ref{claim:fk-cascade} and recompute the function on update, cascading delete, and $\sqcup$.\end{claim}

For our proofs over counter ADTs, we represent increments of a counter $c$ by $inc_i(c)$, where $i$ is a distinct invocation, decrements of $c$ by $dec_i(c)$, and the value of $c$ in database $D$ as $val(c, D) = |\{j \mid inc_j(c) \in D\}| - |\{k \mid dec_k(c) \in D\}|$.

\begin{claim}[Counter ADT increments are \iconfluent with respect to greater-than constraints] \label{claim:ge-inc-adt}
Assume increments are not \iconfluent with respect to some per-counter greater-than constraint $I(D)=val(c, D) < k$, where $k$ is a constant. By definition, there must exist two $I$-$T$-reachable states $S_1$ and $S_2$ with common ancestor reachable by executing write transactions such that $I(S_1) \rightarrow true$ and $I(S_2) \rightarrow true$ but $I(S_1 \sqcup S_2) \rightarrow false$; therefore, $val(c, S_1 \sqcup S_2) \leq k$. However, this implies that $val(c, S_1) \leq k)$, $val(c, S_2$, or both, a contradiction.  \end{claim}

\begin{claim}[Counter ADT increments are not \iconfluent with respect to less-than constraints]
\label{claim:le-inc-adt}
Consider the following transactions:
\begin{align*}
T_{1i}&\coloneqq inc_1(c);~commit\\
T_{2i}&\coloneqq inc_2(c);~commit
\end{align*}
and the less-than inequality constraint:
$$I_i(D) = \{val(c, D) < 2\}$$
$I_i$ holds over the empty database state ($I_i(\{\}) \rightarrow true$) and when $T_a$ and $T_b$ are independently executed:
\begin{align*}
T_{1i}(\{\})&=\{inc_1(c)=1\},~I_i(\{inc_1(c)=1\}) \rightarrow true\\
T_{2i}(\{\})&=\{inc_2(c)\},~I_i(\{inc_2(c)\}) \rightarrow true
\end{align*}
However, merging these states results in invalid state:
$$I_u(\{inc_1(c)\}\sqcup \{inc_2(c)\}) \rightarrow false$$
Therefore, $\{T_{1i}, T_{2i}\}$ is not \iconfluent under $I_u$.
\end{claim}

\begin{claim}[Counter ADT decrements are not \iconfluent with respect to greater-than constraints] The proof is similar to the proof of Claim~\ref{claim:le-inc-adt}; substitute $dec$ for $inc$ and choose $I_i(D) = \{val(c,D) > -2\}$.  \end{claim}

\begin{claim}[Counter ADT decrements are \iconfluent with respect to less-than constraints] \label{claim:le-inc-adt} Unsurprisingly, the proof is almost identical to the proof of Claim~\ref{claim:ge-inc-adt}, but with $<$ instead of $>$ and $dec$ instead of $inc$. \end{claim}

We provide proofs for ADT lists; the remainder are remarkably similar. Our implementation of ADT lists in these proofs uses a lexicographic sorting of values to determine list order. Transactions add a version $v$ to list $l$ via $add(v,l)$ and remove it via $del(v,l)$ (where an item is considered contained in the list if it has been added more times than it has been deleted) and access the length of $l$ in database $D$ via $size(l) = |\{k \mid add(k, l) \in D\}|- |\{m \mid del(m, l) \in D\}|$ (note that $size$ of a non-existent list is zero).

\begin{claim}[Modifying a list ADT is \iconfluent with respect to containment constraints]
  Assume ADT list modifications are not \iconfluent with respect to some equality constraint $I(D)=\{add(k, l) \in D \wedge del(k, l) \notin D \}$ for some constant $k$. By definition, there must exist two $I$-$T$-reachable states $S_1$ and $S_2$ with common ancestor reachable by list modifications such that $I(S_1) \rightarrow true$ and $I(S_2) \rightarrow true$ but $I(S_1 \sqcup S_2) \rightarrow false$; therefore, $add(k, l) \notin \{S_1 \sqcup S_2\}$ or $del(k, l) \in \{S_1 \sqcup S_2\}$. In the former case, neither $S_1$ nor $S_2$ contain $add(k,l)$ a contradiction. In the latter case, if either of $S_1$ or $S_2$ contains $del(k,l)$, it will be invalid, a contradiction.\end{claim}

\begin{claim}[Modifying a list ADT is \iconfluent with respect to non-containment constraints]
Assume ADT list modifications are not \iconfluent with respect to some non-containment constraint $I(D)=\{add(k, l) \notin D \wedge del(k, l) \in D\}$ for some constant $k$. By definition, there must exist two $I$-$T$-reachable states $S_1$ and $S_2$ with common ancestor reachable via list modifications such that $I(S_1) \rightarrow true$ and $I(S_2) \rightarrow true$ but $I(S_1 \sqcup S_2) \rightarrow false$; therefore, $add(k, l) \in \{S_1 \sqcup S_2\}$ and $del(k,l) \notin \{S_1 \sqcup S_2\}$. But this would imply that $add(k, l) \in S_1$, $add(k, l) \in S_2$, or both (while $del(k,l)$ is in neither), a contradiction.\end{claim}

\begin{claim}[Arbitrary modifications to a list ADT are not \iconfluent with respect to equality constraints on the size of the list]
Consider the following transactions:
\begin{align*}
T_{1l}&\coloneqq del(x_i, l);~add(x_a, l);~commit\\
T_{2l}&\coloneqq del(x_i, l);~add(x_b, l);~commit
\end{align*}
and the list size invariant:
$$I_l(D) = \{size(l) = 1\}$$
Now, the size invariant holds on a list of size one ($I_u(\{add(x_i, l)\}) \rightarrow true$) and on independent state modifications:
\begin{align*}
T_{1l}(\{add(x_i, l)\})&=\{add(x_i, l),~del(x_i, l),~add(x_a, l)\}\\
T_{2l}(\{add(x_i, l)\})&=\{add(x_i, l),~del(x_i, l),~add(x_b,1)\}
\end{align*}
However, merging these states result in an invalid state:
\begin{align*}
I_l(&\{add(x_i, l),~del(x_i, l),~add(x_a, l)\}\\\sqcup~&\{add(x_i, l),~del(x_i, l),~add(x_b, l)\})\rightarrow false
\end{align*}
Therefore, $\{T_{1l}, T_{2l}\}$ is not \iconfluent under $I_u$.
\end{claim}

Note that, in our above list ADT, modifying the list is \iconfluent with respect to constraints on the head and tail of the list but not intermediate elements of the list! That is, the head (resp. tail) of the merged list will be the head (tail) of one of the un-merged lists. However, the second element may come from either of the two lists.

\fi

\end{document}